\documentclass[smallextended]{svjour3}

\usepackage{times}
\usepackage[english]{babel}

\makeatletter
\newif\if@restonecol
\makeatother

\let\chapter\section 

\usepackage[ruled, vlined, linesnumbered, nofillcomment]{algorithm2e}

\usepackage[tight,footnotesize]{subfigure}

\usepackage{booktabs}
\usepackage{url}
\usepackage{cite}
\usepackage{xspace}
\usepackage{amsmath}
\usepackage{amssymb}
\usepackage{mdwlist}
\usepackage{graphicx}
\usepackage{rotating}
\usepackage{microtype}
\usepackage[export]{adjustbox}
\usepackage{tikz}
\usetikzlibrary{snakes,arrows,shapes,positioning,fit,calc}
\usepackage{pgffor}
\usepackage{pgfplots}
\usepackage{pgfplotstable}
\usepackage[round]{natbib}
\usepackage{ifpdf}

\usepackage[pdftitle={How to corner your data},
            pdfauthor={Nikolaj Tatti},
			citecolor = black, linkcolor = black, colorlinks = true, bookmarksopen = true, bookmarksnumbered = true, hyperfootnotes = false]{hyperref}

\usepgfplotslibrary{external}
\newcommand\comment[1]{}

\newcommand{\set}[1]{\left\{#1\right\}}
\newcommand{\pr}[1]{\left(#1\right)}
\newcommand{\fpr}[1]{\mathopen{}\left(#1\right)}

\newcommand{\fspr}[1]{\mathopen{}\left[#1\right]}

\newcommand{\abs}[1]{{\left|#1\right|}}

\newcommand{\enpr}[2]{\pr{#1 ,\ldots , #2}}

\newcommand{\funcdef}[3]{{#1}:{#2} \to {#3}}
\newcommand{\score}[1]{s\fpr{#1}}
\newcommand{\sm}{\setminus}

\newcommand{\define}{\leftarrow}

\newcommand{\sgm}[1]{\mathcal{#1}}

\newcommand{\freq}[1]{\mathit{a}\fpr{#1}}
\newcommand{\cs}[1]{\mathit{cs}\fpr{#1}}
\newcommand{\minc}[1]{\mathit{left}\fpr{#1}}
\newcommand{\maxc}[1]{\mathit{right}\fpr{#1}}

\newcommand{\brd}[1]{\mathit{brd}\fpr{#1}}

\newcommand{\mean}[2]{\operatorname{E}_{#1}\fspr{#2}}

\newcommand{\var}[2]{\operatorname{Var}_{#1}\fspr{#2}}

\SetKw{False}{false}
\SetKw{True}{true}
\SetKw{None}{none}
\SetKw{Continue}{continue}
\SetKw{AND}{and}
\SetKw{OR}{or}
\SetKw{EACH}{each}

\SetKwInOut{Input}{input}
\SetKwInOut{Output}{output}
\SetArgSty{textnormal}

\newcommand{\findorder}{\textsc{FindOrder}\xspace}

\pgfdeclarelayer{background}
\pgfdeclarelayer{foreground}
\pgfsetlayers{background,main,foreground}

\definecolor{yafaxiscolor}{rgb}{0.3, 0.3, 0.3}

\definecolor{yafcolor1}{rgb}{0.4, 0.165, 0.553}
\definecolor{yafcolor2}{rgb}{0.949, 0.482, 0.216}
\definecolor{yafcolor3}{rgb}{0.47, 0.549, 0.306}
\definecolor{yafcolor4}{rgb}{0.925, 0.165, 0.224}
\definecolor{yafcolor5}{rgb}{0.141, 0.345, 0.643}
\definecolor{yafcolor6}{rgb}{0.965, 0.933, 0.267}
\definecolor{yafcolor7}{rgb}{0.627, 0.118, 0.165}
\definecolor{yafcolor8}{rgb}{0.878, 0.475, 0.686}

\newlength{\yafaxispad}
\newlength{\yaftlpad}
\newlength{\yaflabelpad}
\newlength{\yafaxiswidth}
\newlength{\yafticklen}
\makeatletter
\def\pgfplots@drawtickgridlines@INSTALLCLIP@onorientedsurf#1{}
\makeatother

\newcommand{\yafdrawaxis}[4]{
	\pgfplotstransformcoordinatex{#1}\let\xmincoord=\pgfmathresult 
	\pgfplotstransformcoordinatex{#2}\let\xmaxcoord=\pgfmathresult 
	\pgfplotstransformcoordinatey{#3}\let\ymincoord=\pgfmathresult 
	\pgfplotstransformcoordinatey{#4}\let\ymaxcoord=\pgfmathresult 
	\pgfsetlinewidth{\yafaxiswidth} 
	\pgfsetcolor{yafaxiscolor}
	\pgfpathmoveto{\pgfpointadd{\pgfpointadd{\pgfplotspointrelaxisxy{0}{0}}{\pgfqpointxy{\xmincoord}{0}}}{\pgfqpoint{-0.5\yafaxiswidth}{\yafaxispad}}}
	\pgfpathlineto{\pgfpointadd{\pgfpointadd{\pgfplotspointrelaxisxy{0}{0}}{\pgfqpointxy{\xmaxcoord}{0}}}{\pgfqpoint{0.5\yafaxiswidth}{\yafaxispad}}}
	\pgfpathmoveto{\pgfpointadd{\pgfpointadd{\pgfplotspointrelaxisxy{0}{0}}{\pgfqpointxy{0}{\ymincoord}}}{\pgfqpoint{\yafaxispad}{-0.5\yafaxiswidth}}}
	\pgfpathlineto{\pgfpointadd{\pgfpointadd{\pgfplotspointrelaxisxy{0}{0}}{\pgfqpointxy{0}{\ymaxcoord}}}{\pgfqpoint{\yafaxispad}{0.5\yafaxiswidth}}}
	\pgfusepath{stroke}
}

\pgfplotscreateplotcyclelist{yaf}{%
{yafcolor1,mark options={scale=0.75},mark=o}, 
{yafcolor2,mark options={scale=0.75},mark=square},
{yafcolor3,mark options={scale=0.75},mark=triangle},
{yafcolor4,mark options={scale=0.75},mark=o},
{yafcolor5,mark options={scale=0.75},mark=o},
{yafcolor6,mark options={scale=0.75},mark=o},
{yafcolor7,mark options={scale=0.75},mark=o},
{yafcolor8,mark options={scale=0.75},mark=o}} 

\pgfplotsset{
	axis y line=left, axis x line=bottom,
	tick align=outside,
	tickwidth=\yafticklen,
	clip = false,
    x axis line style= {-, line width = 0pt, color=black!0},
    y axis line style= {-, line width = 0pt, color=black!0},
    x tick style= {line width = \yafaxiswidth, color=yafaxiscolor, yshift = \yafaxispad},
    y tick style= {line width = \yafaxiswidth, color=yafaxiscolor, xshift = \yafaxispad},
    x tick label style = {font=\scriptsize, yshift = \yaftlpad},
    y tick label style = {font=\scriptsize, xshift = \yaftlpad},
    every axis y label/.style = {at = {(ticklabel cs:0.5)}, rotate=90, anchor=center, font=\scriptsize, yshift = -\yaflabelpad},
    every axis x label/.style = {at = {(ticklabel cs:0.5)}, anchor=center, font=\scriptsize, yshift = \yaflabelpad},
    x tick label style = {font=\scriptsize, yshift = 1pt},
    grid = major,
    major grid style  = {dash pattern = on 1pt off 3 pt},
	every axis plot post/.append style= {line width=\yafaxiswidth} ,
	legend cell align = left,
	legend style = {inner sep = 1pt, cells = {font=\scriptsize}},
	legend image code/.code={%
		\draw[mark repeat=2,mark phase=2,#1] 
		plot coordinates { (0cm,0cm) (0.15cm,0cm) (0.3cm,0cm) };%
	} 
}

\pgfplotsset{colormap={cool}{rgb255(0cm)=(255,255,255); color(1cm)=(yafcolor5); color(3cm)=(black)}}
\pgfplotsset{colormap={cooler}{rgb255(0cm)=(255,255,255); color(1cm)=(yafcolor5)}}

\colorlet{cornercolor}{yafcolor5}
\colorlet{cornerfillcolor}{yafcolor2}

\newcommand{\drawdata}[1]{
\addplot[scatter,forget plot, only marks, mark = *, scatter src = explicit symbolic]
	table[x expr = {-\thisrowno{0}}, y expr = {\thisrowno{1}}, meta index = 2, header = false] {#1};
\addplot[scatter,forget plot, only marks, mark = *, scatter src = explicit symbolic]
	table[x expr = {-\thisrowno{1}}, y expr = {\thisrowno{0}}, meta index = 2, header = false] {#1};
}

\newcommand{\drawvarea}[3]{
\addplot[#3, forget plot, line width = 0.15pt, const plot, fill = #3!50]
	table[y expr = {\coordindex}, x expr = {-\thisrowno{#2} - 1}, header = false] {#1} -- (axis cs: -1, 0) \closedcycle;
\addplot[#3, forget plot, line width = 0.2pt, const plot] table[x expr = {-\coordindex}, y expr = {\thisrowno{#2}}, header = false] {#1};
}

\newcommand{\drawvalues}[2]{
\addplot[scatter,forget plot, only marks, mark = square*, mark size = 0.3pt,
         scatter src = explicit,
		 colormap name = #2,
		 mark options = {line width = 0pt}, scatter/use mapped color= {fill=mapped color, draw = mapped color}]
	table[x expr = {-\thisrowno{0}}, y expr = {\thisrowno{1}}, meta index = 3, header = false] {#1};
}

\newcommand{\drawvalueslower}[2]{
\addplot[scatter,forget plot, only marks, mark = square*, mark size = 0.3pt,
         scatter src = explicit,
		 colormap name = #2,
		 mark options = {line width = 0pt}, scatter/use mapped color= {fill=mapped color, draw = mapped color}]
	table[x expr = {-\thisrowno{1}}, y expr = {\thisrowno{0}}, meta index = 3, header = false] {#1};
}

\newcommand{\drawgraphvalue}[2]{
\addplot[yafcolor3, forget plot, line width = 0.15pt, const plot, fill = yafcolor3!50]
	table[y expr = {\coordindex}, x expr = {-\thisrowno{1} - 1}, header = false] {#1} -- (axis cs: -1, 0) \closedcycle;
\addplot[yafcolor2, forget plot, line width = 0.15pt, const plot, fill = yafcolor2!50]
	table[y expr = {\coordindex}, x expr = {-\thisrowno{0} - 1}, header = false] {#1} -- (axis cs: -1, 0) \closedcycle;

\addplot[scatter,forget plot, only marks, mark = square*, mark size = 0.2pt,
         scatter src = explicit,
		 colormap name = cool,
		 mark options = {line width = 0pt}, scatter/use mapped color= {fill=mapped color, draw = mapped color}]
	table[x expr = {-\thisrowno{0}}, y expr = {\thisrowno{1}}, meta index = 3, header = false] {#2};

\addplot[yafcolor3, forget plot, line width = 0.2pt, const plot] table[x expr = {-\coordindex}, y expr = {\thisrowno{1}}, header = false] {#1};
\addplot[yafcolor2, forget plot, line width = 0.2pt, const plot] table[x expr = {-\coordindex}, y expr = {\thisrowno{0}}, header = false] {#1};

}

\newcommand{\drawgraph}[2]{
\addplot[yafcolor2, forget plot, line width = 0.15pt, const plot, fill = yafcolor2!50]
	table[y expr = {\coordindex}, x expr = {-\thisrowno{1} - 1}, header = false] {#1} -- (axis cs: -1, 0) \closedcycle;
\addplot[yafcolor5, forget plot, line width = 0.15pt, const plot, fill = yafcolor5!50]
	table[y expr = {\coordindex}, x expr = {-\thisrowno{0} - 1}, header = false] {#1} -- (axis cs: -1, 0) \closedcycle;

\addplot[yafcolor2, forget plot, line width = 0.15pt, const plot] table[x expr = {-\coordindex}, y expr = {\thisrowno{1}}, header = false] {#1};
\addplot[yafcolor5, forget plot, line width = 0.15pt, const plot] table[x expr = {-\coordindex}, y expr = {\thisrowno{0}}, header = false] {#1};

\addplot[scatter,forget plot, only marks, mark = *, mark size = 0.2pt, scatter src = explicit symbolic]
	table[x expr = {-\thisrowno{0}}, y expr = {\thisrowno{1}}, meta index = 2, header = false] {#2};
}

\newcommand{\drawgraphiv}[2]{
\addplot[yafcolor3, forget plot, line width = 0.15pt, const plot, fill = yafcolor3!50]
	table[y expr = {\coordindex}, x expr = {-\thisrowno{2} - 1}, header = false] {#1} -- (axis cs: -1, 0) \closedcycle;
\addplot[yafcolor2, forget plot, line width = 0.15pt, const plot, fill = yafcolor2!50]
	table[y expr = {\coordindex}, x expr = {-\thisrowno{1} - 1}, header = false] {#1} -- (axis cs: -1, 0) \closedcycle;
\addplot[yafcolor5, forget plot, line width = 0.15pt, const plot, fill = yafcolor5!50]
	table[y expr = {\coordindex}, x expr = {-\thisrowno{0} - 1}, header = false] {#1} -- (axis cs: -1, 0) \closedcycle;

\addplot[yafcolor3, forget plot, line width = 0.15pt, const plot] table[x expr = {-\coordindex}, y expr = {\thisrowno{2}}, header = false] {#1};
\addplot[yafcolor2, forget plot, line width = 0.15pt, const plot] table[x expr = {-\coordindex}, y expr = {\thisrowno{1}}, header = false] {#1};
\addplot[yafcolor5, forget plot, line width = 0.15pt, const plot] table[x expr = {-\coordindex}, y expr = {\thisrowno{0}}, header = false] {#1};

\addplot[scatter,forget plot, only marks, mark = *, mark size = 0.2pt, scatter src = explicit symbolic]
	table[x expr = {-\thisrowno{0}}, y expr = {\thisrowno{1}}, meta index = 2, header = false] {#2};
}

\newcommand{\drawpoints}[3]{
\foreach \x/\y [count=\xi] in {#3} {
\node[#1] (#2\xi) at (\x,\y) {};
}}

\newcommand{\drawgrid}[3]{
\foreach \x in {#1 ,..., #2} { \draw[#3] (#1, \x) -- (\x, \x);}
\foreach \x in {#1 ,..., #2} { \draw[#3] (\x, #2) -- (\x, \x);}
\draw[#3] (#1, #1)  -- (#2, #2);
}

\newif\ifapx
\apxtrue

\title{Discovering Bands from Graphs}
\author{Nikolaj Tatti}

\institute{
N. Tatti \at
HIIT, Department of Information and Computer Science, Aalto University, Finland, and\\
Department of Computer Science, KU Leuven, Leuven, Belgium\\
\email{nikolaj.tatti@aalto.fi}}

\journalname{Data Mining and Knowledge Discovery}

\date{}

\begin{document}
\maketitle

\begin{abstract}
Discovering the underlying structure of a given graph is one of the fundamental
goals in graph mining.
Given a graph, we can often order vertices in a way that neighboring vertices
have a higher probability of being connected to each other. This implies that
the edges form a band around the diagonal in the adjacency matrix.
Such structure may rise for example if the graph was created over time: each vertex
had an active time interval during which the vertex was connected with other active
vertices.

The goal of this paper is to model this phenomenon. To this end, we formulate
an optimization problem: given a graph and an integer $K$, we want to order
graph vertices and partition the ordered adjacency matrix into $K$ bands such
that bands closer to the diagonal are more dense.  We measure the goodness of a
segmentation using the log-likelihood of a log-linear model, a flexible family
of distributions containing many standard distributions.  We divide the problem
into two subproblems: finding the order and finding the bands.  We show that
discovering bands can be done in polynomial time with isotonic regression, and
we also introduce a heuristic iterative approach.  For discovering
the order we use Fiedler order accompanied with a simple combinatorial
refinement.  We demonstrate empirically that our heuristic works well in
practice.

\end{abstract}

\section*{Categories and Subject Descriptors}
H.2.8\,[\textbf{Database\,management}]:\,Database\,applications--\textit{Data\,mining}



\section{Introduction}\label{sec:intro}

Consider a dataset given in Figure~\ref{fig:paleotoy:a}. This data contains
139 species discovered at 501 sites~\citep{fortelius05now}. As different
species live in different eras, the dataset can be sorted\footnote{Here, we sorted the data using the Fiedler order, see Section~\ref{sec:order}.} such that the data points
form a band. Let us construct a similarity matrix between the sorted species,
where the weight between two species is the number of sites. Since a large
number of the species-pairs will have do not share a single site, it is beneficial to view the matrix
as a weighted graph. We see from the graph, given in Figure~\ref{fig:paleotoy:b},
that most of the edges will be located close to the diagonal, forming a band.

\begin{figure}[htb!]\hfill
\subfigure[Data matrix\label{fig:paleotoy:a}]{
\includegraphics[width=6cm]{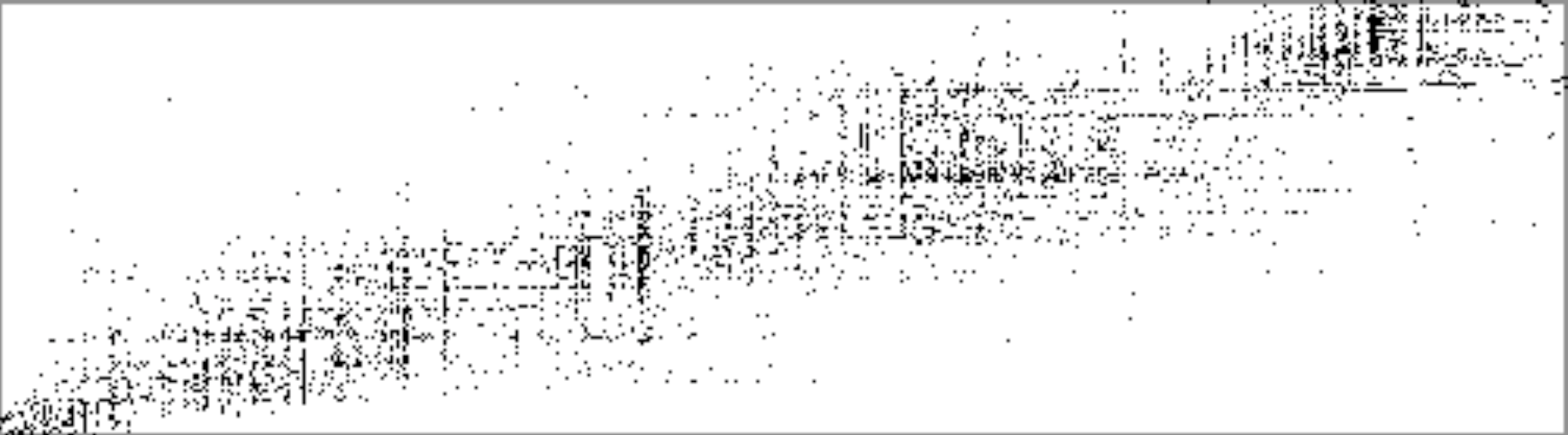}
}\hfill
\subfigure[Similarity graph\label{fig:paleotoy:b}]{
\begin{tikzpicture}
\begin{axis}[width = 4cm,height = 4cm, ticks = none, grid=none,
    every axis plot post/.style={}, ]

\drawvalues{data/paleo-dist.ord}{cool}
\drawvalueslower{data/paleo-dist.ord}{cool}

\end{axis}
\end{tikzpicture}}\hfill
\subfigure[Discovered bands\label{fig:paleotoy:c}]{
\begin{tikzpicture}
\begin{axis}[width = 4cm,height = 4cm, ticks = none, grid=none,
    every axis plot post/.style={}]

\drawgraphvalue{data/paleo-dist.out}{data/paleo-dist.ord}

\end{axis}
\end{tikzpicture}}\hspace*{\fill}
\caption{139 different species discovered from 501 different paleontological sites}
\label{fig:paleotoy}
\end{figure}

The phenomenon of having edges near the diagonal is not uncommon. For example, assume that the
graph was constructed over time and that each vertex had an active
time interval during which it connected to other active vertices with a higher
probability. In such case, we should be able to arrange the vertices such that
the edges are concentrated around the diagonal. As another example consider a graph
with a(n overlapping) clustering structure. If we can rearrange clusters such that
only neighboring clusters have significant overlap, if any, then there exists a
vertex order such that the edges are close to the diagonal.

The goal of this paper is to quantify this phenomenon, see
Figure~\ref{fig:paleotoy:c} for an example.  More formally, we introduce the following
optimization problem. Given a graph and an integer $K$, order vertices and
segment the adjacency matrix into $K$ bands such that a) each band respects the vertex order---when
drawn the boundary of the band must either move down or right, b) the edge
density in inner bands is higher than the density in the outer bands, and c)
segments are as homogeneous as possible, according to some score function. 
Note that these bands may have varying thickness as illustrated in
Figure~\ref{fig:paleotoy:c}.

In order to score the bands, we model the graph as a mixture model, where each
band is a Erd\H{o}s-R\'enyi model~\citep{erdos1960erg}. Our goal is to find
bands that optimize the likelihood of this model. As an application, discovered
bands can be used for determining communities for individual vertices: if $(u,
v)$ belongs to the $k$th band, then we say that $v$ belongs to the $k$th
community of $u$, small $k$s corresponding to the inner circles. 

We break this optimization problem into two natural subproblems. The first
problem is to find optimal $K$ bands given the order and the second problem is
to find a good order. We approach the latter problem by using 
Fiedler order~\citep{fiedler75fiedler} accompanied with a simple greedy refinement heuristic.

Most of this work is devoted into solving the first subproblem which happens to
have a polynomial solution. In fact, this problem resembles a monotonic
segmentation problem~\citep[see][]{haiminen:04:unimodal}, however, it is much
more intricate and there is no obvious technique for solving such problem. We
will show that for certain scores, we do not have to consider all possible
segmentations. We introduce a concept of \emph{borders}. Roughly speaking, a
border divides the adjacency matrix into two parts such that the inner part has a higher
average. We will show that the optimal segmentation can be constructed from the
borders.  This allows us to transform the original optimization problem into
two separate problems. First we need to discover all the borders, secondly we
need to select the optimal segmentation using borders as candidates.
Surprisingly, the second subproblem turns out to be an instance of the sequence
segmentation problem, and can be solved using a standard dynamic
program given by~\citet{bellman:61:on}.

We consider two techniques for discovering borders. The first approach is based
on isotonic regression~\citep{spouge:03:isotonic}, and gives us an exact
solution.  Despite being a polynomial-time solution, this approach requires
that the graph is stored in a full form. Hence, we also present an iterative
heuristic technique that uses sparsity of the graph to its advantage.

The rest of the paper is organized as follows. We introduce preliminaries and
formally state our problem in Section~\ref{sec:prel}. We introduce the concept of
borders in Section~\ref{sec:borders} and present algorithms for discovering borders
in Section~\ref{sec:discovery}. We consider discovering orders in Section~\ref{sec:order}.
We present related work in Section~\ref{sec:related}
and experimental evaluation in Section~\ref{sec:experiments}. Finally, we present
our conclusions in Section~\ref{sec:conclusion}.

\section{Preliminaries and Problem Statement}
\label{sec:prel}
In this section we present our notation and give the formal problem statement.
We first introduce the optimization problem for graphs and then cast this problem
into a more general setup.

\subsection{Discovering Bands from Graphs}

Our first task is to define formally what we mean by a band.  In order to do
this, assume an undirected graph $H = (V, F)$. If we are given an order $o$ on vertices,
essentially a mapping $\funcdef{o}{1, \ldots, \abs{V}}{V}$, we say that $H$
respects the order if the neighborhood of each vertex can be seen as a segment
w.r.t. the order, that is, for every $v \in V$, there exist integers $s$ and $e$ such that $\set{u \in
V \mid (v, u) \in F} \cup \set{v} = \set{o(i) \mid s \leq i \leq e}$.
If we order the vertices according to $o$, then $H$ will have all its edges
next to diagonal. Our goal is given a graph $G$, find $o$ and $H$ optimizing
a certain score.

Let us now define the score that we wish to optimize. Assume that we are given a graph $G = (V, E)$.
Let $X \subseteq V \times V$ be a subset of vertex pairs. Let us define
\[
	\score{X} = \abs{X \cap E}\log \freq{X} + \abs{X \setminus E} \log(1 - \freq{X}),
	\quad\text{where}\quad  \freq{X} = \frac{\abs{X \cap E}}{\abs{X}},
\]
which is a maximum log-likelihood of a Bernoulli variable. 
We can now formulate the optimization problem.

\begin{problem}[2-band discovery]
Given a graph $G = (V, E)$ find an order $o$ on vertices and a graph $H$ respecting that
order $o$ and maximizing $\score{E(H)} + \score{(V \times V) \setminus E(H)}$ such that
$\freq{E(H)} \geq \freq{(V \times V) \setminus E(H)}$.
\end{problem}
In other words, we are modelling edges in $G$ as a mixture of two Bernoulli variables.
The last constraint requires that the density of $G$ in the edges of $H$, that is,
next to diagonal, should be higher than in the non-edges of $H$. As mentioned in the introduction,
we are interested in a more general setup where we can discover several bands.
This gives us the following optimization problem.

\begin{problem}[$K$-band discovery]
Given a graph $G = (V, E)$ and integer $K$, find an order $o$ on vertices and
$K + 1$ graphs $H_0, \ldots H_K$ respecting order such that $\emptyset = E(H_0)
\subsetneq \cdots \subsetneq E(H_K) = V \times V$, the density is decreasing,
$\freq{C_i} \geq \freq{C_{i + 1}}$,
and the score $\sum_{i = 1}^K \score{C_i}$ is maximized, where $C_i = E(H_i) \setminus E(H_{i - 1})$. 
\end{problem}

In order to approach this optimization problem we will split it into two subproblems.
The first subproblem is to find the bands for a fixed order.

\begin{problem}[ordered $K$-band discovery]
\label{prob:orderband}
Given a graph $G = (V, E)$ an integer $K$ and an order $o$ on vertices, find
$K + 1$ graphs $H_0, \ldots H_K$ respecting order such that $\emptyset = E(H_0)
\subsetneq \cdots \subsetneq E(H_K) = V \times V$, the density is decreasing,
$\freq{C_i} \geq \freq{C_{i + 1}}$,
and the score $\sum_{i = 1}^K \score{C_i}$ is maximized, where $C_i = E(H_i) \setminus E(H_{i - 1})$. 
\end{problem}

The second subproblem is to find the actual order. Our main focus will be the
first subproblem for which we develop a polynomial exact solution. We address
discovering the order in Section~\ref{sec:order}.

\subsection{Band Discovery as Monotonic 2D-segmentation}

Graph $H$ in ordered band discovery has a special property: if we order $H$
based on the order $o$ and consider the upper-half of the adjacency matrix,
then we see that all 1s are concentrated next to the diagonal. We will use this
observation to cast band discovery into a more general segmentation problem.

In order to do so,
assume that we are given a dataset of size $M \times N$.  Define $A = \set{(a,
b) \mid 1 \leq a \leq M, 1 \leq b \leq N}$ to be the set of all entries.
We say that $U \subseteq A$ is a \emph{corner} if for every $(a, b) \in U$,
and for every $x$ and $y$ such that $1 \leq x \leq a$ and $1 \leq y \leq b$, 
an entry $(x, y)$ is a member of $U$.

Given an integer $K$ and a corner $B$, we define a $K$-\emph{segmentation} to
be the set of $K + 1$ corners, $U_0, \ldots, U_K$ such that $U_0 = B$,
$U_K = A$, and $U_{i - 1} \subseteq U_i$ for each $i = 1, \ldots, K$. We will
refer to the difference set $U_i \sm U_{i - 1}$ as a \emph{segment}.

Our goal is to segment given data into $K$ segments such that this segmentation
maximizes the likelihood of a log-linear model for each segment. By log-linear
models, also known as exponential family, we mean models whose probability
density function can be written as
\[
	p(x \mid r) = \exp(q(x) + Z(r) + rS(x)),
\]
where $S$ is a function mapping each data point to a real number, $r$ is the
parameter of the model, and $Z(r)$ is the normalization constant.  Many
standard distributions such as Bernoulli, binomial, Gaussian, and Poisson
are log-linear distributions. Interestingly, one can show with a direct
computation that using a Gaussian distribution with a fixed variance corresponds
to minimizing $L_2$ error.

Before we define our score, let us demonstrate that we can safely assume that
$S(x) = x$ and $q(x) = 0$. In order to do that, let $C_1, \ldots, C_K$ be $K$
subsets of $A \setminus B$ such that $C_1 \cup \cdots \cup C_K = A \setminus B$ and $C_i \cap C_j =
\emptyset$ for $i \neq j$.
Assume that we have assigned to each $C_i$, a parameter $r_i$. 

We will measure the goodness of a segmentation by the log-likelihood of the model,
\[
\begin{split}
	\sum_{i = 1}^K \sum_{c \in C_i} \log p(D(c) \mid r_i) 
	&= \sum_{i = 1}^K\big( \abs{C_i}Z(r) + \sum_{c \in C_i} q(D(c)) + \sum_{c \in C_i} r_iS(D(c))\big) \\
	&= \sum_{c \in A \setminus B} q(D(c)) + \sum_{i = 1}^K\big( \abs{C_i}Z(r) + r_i\sum_{c \in C_i} S(D(c))\big)
	\quad.
\end{split}
\]
Note that the first term does not depend on $C_i$ or $r_i$. Consequently, we can ignore
it and by doing so ignore $q(x)$ term. We can also safely assume that $S(x) = x$. Otherwise,
we can transform data $D$ to a new dataset $D'$ by setting $D'(c) = S(D(c))$.

We can now formally define our score. Given a segment $C$ and a parameter $r$,
we define the score $\score{C \mid r}$ as 
\[
	\score{C \mid r} = \abs{C}(Z(r) + r\freq{C}),\ \text{where}\  \freq{C} = \frac{1}{\abs{C}}\sum_{c \in C} D(c)\quad.
\]
We also define $\score{C} = \sup_r \score{C \mid r}$ to be the score of the
optimal model. Given a $K$-segmentation $\sgm{U} = \enpr{U_0}{U_K}$,
we define the score $\score{\sgm{U}}$ as
\[
	\score{\sgm{U}} = \sum_{i = 1}^K \score{U_i \sm U_{i - 1}}\quad.
\]

We say that a $K$-segmentation $\sgm{U}$ is \emph{monotonic} if $\freq{U_{i + 1}} \leq \freq{U_i}$
for each $1 \leq i \leq K - 1$ such that $U_i \neq \emptyset$. Our goal is to solve
the following problem.

\begin{problem}[2D-segmentation]
\label{prob:segmentation}
Given a dataset $D$, a corner $B$, a log-linear model, and an integer $K$, find a monotonic
$K$-segmentation $\sgm{U}$ maximizing $\score{\sgm{U}}$.
\end{problem}

We can now see that band discovery is in fact an instance of 2D-segmentation.
The dataset $D$ is in fact the adjacency matrix of $G$, the corner $B$ is a
diagonal, and the score model is Bernoulli variable. We should point out that
solving the more general problem has its advantages. If we have weights on edges,
we can use some other log-linear model, such as Poisson model, to score the
segmentation.  Moreover, we do not have to restrict ourselves to graphs, we can
segment any given matrix. On the other hand, discovering ordered bands is essentially
as difficult as solving 2D-segmentation, that is, the amount of additional work we need to
do to solve the more general case is negligible.

From now on, we will ignore $B$ for the sake of simplicity, and assume that we
want to segment the whole dataset. The algorithms that we present can be easily
adjusted to handle the more general case when $B$ is given.

The next two sections are devoted to solving Problem~\ref{prob:segmentation}. We discuss
discovering the vertex order in Section~\ref{sec:order}.

\section{Borders}
\label{sec:borders}

We can easily show that there are ${N + M \choose M}$ different corners for a
data of size $M \times N$.  This implies that we cannot enumerate all possible
corners in order to solve Problem~\ref{prob:segmentation}. 
However, we can show that we do not have to consider all possible corners.

In order to do so,
we say that a corner $U$ is a \emph{border} if there are no corners $X$ and $Y$
such that $X \subsetneq U \subsetneq Y$ and $\freq{Y \sm U}
\geq \freq{U \sm X}$.  We denote all borders by $\brd{D}$. An example of
a non-border is given in Figure~\ref{fig:toy:a}.

\begin{figure}[htb!]
\hfill
\subfigure[\label{fig:toy:a}]{
\begin{tikzpicture}[baseline]
\begin{axis}[
    width = 3.5cm,height = 3.5cm,
    xmax = 29, ymin = -29, ymax = 0, xmin = 0,
    no markers, ticks = none, grid=none,
	legend entries = {$X$, $Y$, $Z$},
	legend pos = outer north east ]

\addplot[fill, cornerfillcolor!40, forget plot] coordinates { (29, -29) (29, 0) (0, 0) (29, -29) };
\addplot+[fill, cornerfillcolor!10, forget plot] table[x expr = \coordindex, y expr = {-\thisrowno{0} - \thisrowno{1}}, header = false] {toy1.dat} \closedcycle;

\addplot+[fill, cornerfillcolor!60, forget plot] table[x expr = \coordindex, y expr = {-\thisrowno{0}}, header = false] {toy1.dat} \closedcycle;
\addplot+[fill, cornerfillcolor!20, forget plot] table[x expr = \coordindex, y expr = {-\thisrowno{0} - \thisrowno{2}}, header = false] {toy1.dat} \closedcycle;

\addplot[cornerfillcolor!50, forget plot]  coordinates { (29, -29) (29, 0) (0, 0) } \closedcycle;

\addplot+[cornercolor] table[x expr = \coordindex, y expr = {-\thisrowno{0}}, header = false] {toy1.dat} -- (axis cs: 29, -29) -- (axis cs: 0, 0) \closedcycle;
\addplot+[cornercolor, dashed] table[x expr = \coordindex, y expr = {-\thisrowno{0} - \thisrowno{1}}, header = false] {toy1.dat};
\addplot+[cornercolor, dotted] table[x expr = \coordindex, y expr = {-\thisrowno{0} - \thisrowno{2}}, header = false] {toy1.dat};

\end{axis}
\end{tikzpicture}}
\hfill
\subfigure[\label{fig:toy:b}]{
\begin{tikzpicture}[baseline]
\begin{axis}[
    width = 3.5cm,height = 3.5cm,
    xmax = 29, ymin = -29, ymax = 0, xmin = 0,
    no markers, ticks = none, grid=none,
	legend entries = {$X$, $Y$},
	legend pos = outer north east]

\addplot[fill, cornerfillcolor!40, forget plot] coordinates { (29, -29) (29, 0) (0, 0) (29, -29) };

\addplot+[fill, cornerfillcolor!50, forget plot] table[x expr = \coordindex, y expr = {-\thisrowno{1}}, header = false] {toy2.dat} \closedcycle;
\addplot+[fill, cornerfillcolor!20, forget plot] table[x expr = \coordindex, y expr = {-\thisrowno{0}}, header = false] {toy2.dat} \closedcycle;
\addplot+[fill, cornerfillcolor!10, forget plot] table[x expr = \coordindex, y expr = {-min(\thisrowno{1}, \thisrowno{0})}, header = false] {toy2.dat} \closedcycle;

\addplot[cornerfillcolor!50, forget plot]  coordinates { (29, -29) (29, 0) (0, 0) } \closedcycle;

\addplot+[cornercolor, solid] table[x expr = \lineno, y expr = {-\thisrowno{1}}, header = false] {toy2.dat} -- (axis cs: 29, -29) -- (axis cs: 0, 0) \closedcycle;
\addplot+[cornercolor, dashed] table[x expr = \lineno, y expr = {-\thisrowno{0}}, header = false] {toy2.dat} -- (axis cs: 29, -29);

\end{axis}
\end{tikzpicture}}
\hspace*{\fill}
\caption{Toy examples of corners. In Figure~\ref{fig:toy:a} $X$ is a corner but not a border since $\freq{Z \sm X} \geq \freq{X \sm Y}$.
Figure~\ref{fig:toy:b} shows an example of Proposition~\ref{prop:chain}. Both $X$ and $Y$ are corners but $X$ cannot be a border
since $\freq{Y \sm X} \geq \freq{X \sm Y}$. The density of the corners is represented by the shade of the color.}

\end{figure}
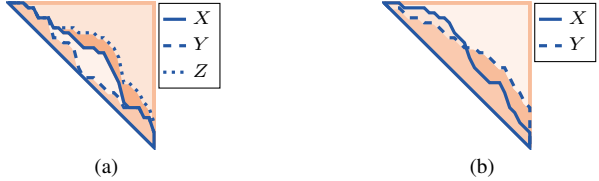

The next key proposition shows that we can safely ignore all corners that
are not borders.

\begin{proposition}
\label{prop:borders}
Let $\sgm{U}$ be a $K$-segmentation. Then there exists a $K$-segmentation $\sgm{V} =
\enpr{V_0}{V_K}$ such that $V_i \in \brd{D}$ for every $i = 0, \ldots, K$ and
$\score{\sgm{V}} \geq \score{\sgm{U}}$.
\end{proposition}

We present the proof for Proposition~\ref{prop:borders} in Appendix~\ref{sec:proofborder}.

The following proposition states
that the borders form a chain, a key property, which we also illustrate in
Figure~\ref{fig:toy:b}.

\begin{proposition}
\label{prop:chain}
Let $D$ be a dataset.  Let $U,\ V \in \brd{D}$ be two borders.
Then either $U \subseteq V$ or $V \subseteq U$.
\end{proposition}

\begin{proof}
Assume that the proposition does not hold, then both sets $U \sm V$ and
$U \sm V$ are non-empty. Assume that $\freq{U \sm V} \geq \freq{V
\sm U}$, otherwise swap $U$ and $V$. Let $X = U \cap V$ and $Y = U \cup
V$.  We have $\freq{Y \sm V} = \freq{U \sm V} \geq \freq{V
\sm U} = \freq{V \sm X}$.  Consequently, $V$ is not a border, which
is a contradiction.
\qed\end{proof}

Proposition~\ref{prop:chain} has several key implications.  Assume that we have
a dataset $D$ of size $N \times M$. Then the number of borders is bounded by
$NM + 1$. From now on we will assume that $\brd{D} = U_0, \ldots, U_L$ is
ordered from smallest to largest. In order to store this list efficiently, we
define $C_i = U_i \sm U_{i - 1}$ for $i = 1, \ldots, L$. Instead of storing
$U_i$ we simply store $C_i$. Since $C_i$ are disjoint and the union of $C_i$
is equal to $A$, the set of all entries, we can store the borders in $O(NM)$ space.

Propositions~\ref{prop:borders}--\ref{prop:chain} allow us to divide Problem~\ref{prob:segmentation}
into two subproblems. The first problem is to find all the borders.

\begin{problem}
\label{prob:borders}
Given a dataset $D$, compute $\brd{D}$.
\end{problem}

After we have discovered the borders, we can now use them to find the optimal
segmentation. We will restate the problem in a slightly different manner, using
directly segments instead of corners. In order to do so, given a list of segments
$C_1, \ldots, C_L$, let us define $C_{[a, b]}$ to mean $\bigcup_{i = a}^b C_i$.
Note that if we set $C_i = U_i \sm U_{i - 1}$, then it follows that
$C_{[a, b]} = U_b \sm U_{a - 1}$. This implies that we can reformulate
the optimization problem as follows.

\begin{problem}
\label{prob:seqseg}
Given a sequence of $L$ segments $C_1, \ldots, C_L$ and an integer $K$,
find $K$ intervals, $[b_i, e_i]$, such that $b_1 = 1$, $e_K = L$, 
and $b_i = e_{i - 1} + 1$, for $i = 2, \ldots, K$, optimizing 
\[
	\sum_{i}^{K} \score{C_{[b_i, e_i]}}\quad.
\]
\end{problem}

Note that we have dropped the monotonicity condition from the definition of the
problem.  We will see later in Corollary~\ref{cor:monotone} that if $C_i$ are
constructed from borders, that is, $C_i = U_i \sm U_{i - 1}$, then
$\freq{C_{i + 1}} < \freq{C_i}$. Thus monotonicity will automatically be
guaranteed.

Problem~\ref{prob:seqseg} is in fact an instance of a classic sequence
segmentation problem, where the goal is to split a sequence of length $L$ into $K$
homogeneous segments. This can be solved by a dynamic program in $O(L^2K)$ time
and $O(KL)$ space~\citep{bellman:61:on}.  To see this let us write $O_k(i)$ to be an
optimal $k$-segmentation for a sequence $C_1, \ldots, C_i$.
Then $O_k(i)$ is equal to  $O_{k - 1}(j)$ augmented with $[j +1 , i]$, and $j < i$
is selected such that the score
\[
	\score{O_{k - 1}(j)} + \score{C_{[j + 1, i]}}
\]
is maximized.  We can iteratively compute this
by first computing $O_1(i)$ for each $i = 1,\ldots, L$, and use the above
identity to $O_k(i)$ from $O_{k - 1}(i)$ until we reach $K$ segments.

To summarize, we discover bands in 3 steps:
\begin{enumerate*}
\item Order the dataset, if one is not given. 
\item Compute borders $\brd{D}$ (Problem~\ref{prob:borders}).
\item Segment borders to obtain $K$-segmentation (Problem~\ref{prob:seqseg}).
\end{enumerate*}

We have already shown that Problem~\ref{prob:seqseg} can be solved by using the
classic segmentation technique. In the next two sections we discuss how to
obtain the borders, either exactly or approximately. Finally, in
Section~\ref{sec:order} we discuss how to obtain the order. As discovering the
order seems to be computationally intractable we resort to spectral heuristics.

\section{Discovering borders}
\label{sec:discovery}

In this section we focus on discovering borders (Problem~\ref{prob:borders}).
Namely, we provide a polynomial-time dynamic program that discovers borders
correctly.  In addition we provide a heuristic for the cases when the exact
discovery is too time-consuming.

\subsection{Maximal and Minimal Corners}

In order to define the discovery algorithm, we need to introduce the notion
of maximal and minimal corners. We will show that these notions are closely
related to borders.

Let $U$ be a corner. We define $\minc{U}$ to be the \emph{minimal corner} $V
\subsetneq U$ such that $\freq{U \sm V}$ is the smallest possible. If
there are several candidates, we choose the smallest one.\!\footnote{We can
easily show that this choice is unique.} We also define $\maxc{U}$ to be the
\emph{maximal corner} $V \supsetneq U$ such that $\freq{V \sm U}$ is the
largest possible.  If there are several candidates, we choose the largest one.

We can use maximal and minimal corners to describe borders.

\begin{proposition}
\label{prop:move}
Let $U$ and $V$ be two consecutive
borders. Then $U = \minc{V}$ and $V = \maxc{U}$.
\end{proposition}

We prove this proposition in Appendix~\ref{sec:propmove}.

This proposition has an important corollary that shows
why we can ignore the monotonicity condition in Problem~\ref{prob:seqseg}.
Indeed segments between borders will have automatically
decreasing average.

\begin{corollary}
\label{cor:monotone}
Let $U, V, W \in \brd{D}$ be three consecutive borders, $U \subsetneq V \subsetneq W$.
Then $\freq{V \sm U} > \freq{W \sm V}$.
\end{corollary}

\begin{proof}
Proposition~\ref{prop:move} implies that $\maxc{U} = V$.
Since $V \subsetneq W$, Lemma~\ref{lem:split} (given in Appendix~\ref{sec:propmove}) implies that
$\freq{V \sm U} > \freq{W \sm V}$.
\qed\end{proof}

\subsection{Exact discovery}\label{sec:exact}

In this section we present an algorithm for computing the borders.

Discovering borders exactly is in fact an instance of isotonic regression for a
grid. In this regression we are given a grid, and a set of values on each grid
entry.  The goal is to find another set of values such that they decrease as we
move towards the one selected corner and $L_2$ error is minimized as we move
towards that corner.

\begin{problem}
Let $D$ be a dataset of size $M \times N$. Find a function $f$ such that 
$f(x, y) \geq f(x, y + 1)$ and $f(x, y) \geq f(x + 1, y)$ minimizing the cost
\[
	\sum_{(x, y)} (f(x, y) - D(x, y))^2\quad.
\]
\end{problem}

Once these these values are discovered, we can reconstruct the borders using the following proposition.

\begin{proposition}
Let $D$ be a dataset and let $f$ be the solution to the grid isotonic regression.
Let $B$ be a border. Then there is $\sigma$ such that $B = \set{(x, y) \mid f(x, y) \geq \sigma}$.
\end{proposition}

\begin{proof}
Let $m = \min_{(x, y) \in B} f(x, y)$. Define $Y = \set{(x, y) \notin B \mid
f(x, y) \geq m}$. We need to show that $Y = \emptyset$.  Assume otherwise.
Note that since $f$ is monotonic, $B \cup Y$ is a corner.
Since $B$ is a border, Proposition~\ref{prop:move} implies that $\freq{Y} < m$.
An entry, if such exists outside $B \cup Y$ must be lower than $m$. Hence, we can
decrease the values of $f$ in $Y$ by a small amount, say $\epsilon$, without violating the monotonic
constraint. Let us consider the effect. In order to do this, consider the
following derivative,
\[
	\frac{d}{dc}\sum_{(x, y) \in Y}(f(x, y) - c - D(x, y))^2 =  2\sum_{(x, y) \in Y} -f(x, y) + D(x, y) < 0,
\]
where the inequality follows from the fact that $\freq{Y} < m \leq f(x, y)$ for any $(x, y) \in Y$.
Hence, we can decrease the values of $f$ in $Y$ by a small amount such that the monotonicity still holds
and the score is decreased. This contradicts the fact that $f$ is the optimal solution.
\qed\end{proof}

This proposition gives us a simple approach to discover borders by varying $\sigma$.

Finding the solution for grid isotonic regression can be done in $O((NM)^2)$
time and $O(NM)$ space by an algorithm of~\citet{spouge:03:isotonic}. However,
the time complexity of the algorithm is overly pessimistic. The algorithm runs
in $O((NM)L)$, where $L$ is the number of borders. The number of borders is
$NM$, at worst, but in practice it is much smaller. The algorithm
of~\citet{spouge:03:isotonic} is a conquer-and-divide algorithm. The running
time $O((NM)L)$ is based on the pessimistic assumption that each division is
imbalanced, that is, only one point is separated. If these divisions are
(nearly) balanced, then in practice the running time will be closer to
$O((NM)\log L)$.  Additional speed-ups are possible if instead of considering
the full matrix of size $NM$ we first compute the smallest corner containing
the whole data, and apply the algorithm to the corner. The points that are left
outside constitute a border with $0$ average. This trick may speed-up the search
significantly but it is highly vulnerable to noise as one non-zero point is enough to
counter this optimization.

\subsection{Heuristic discovery}
\label{sec:heuristic}

The computational complexity of isotonic regression may become too high in
practice, especially due to the $O(NM)$ term.
Hence, in this section we present a practical heuristic approach. The idea here
is to sort individual entries of $D$ into a sequence. Once we have this
sequence we can use it to compute candidates for borders. We can then use these
candidates to rearrange the entries again, and repeat the procedure until convergence.

To make this more formal,
assume that we are given a dataset $D$ of size $M \times N$.  Let $T = \enpr{t_1 = (x_1,
y_1)}{t_{NM} = (x_{NM}, y_{NM})}$ be a sequence of all entries of $D$. We say
that $T$ is a \emph{monotonic entry order} if for any $k$, $1 \leq k \leq NM$,
the set $t_1, \ldots, t_k$ is a corner.

Given a monotonic entry order $T = t_1, \ldots, t_{NM}$ and an integer $i$,
we define
\[
	\maxc{i; T} = \set{t_1, \ldots, t_j}, \ \text{where}\  j = \arg \max_{j > i} \freq{t_{i + 1}, \ldots, t_j}\ .
\]
If there are ties, we select the largest index.

Given a monotonic entry order $T$, consider the following process.  Set $U_0 =
\emptyset$ and then iteratively compute $U_i = \maxc{\abs{U_{i - 1}}; T}$ until
we reach $U_L$ containing all the entries.  We will write $\brd{T} = U_0,
\ldots, U_L$.

We say that a monotonic entry order $T = t_1, \ldots, t_{NM}$ is \emph{compatible} with
$\brd{D}$ if for each border $U \in \brd{D}$ there exists an index $k$ such
that $t_1, \ldots, t_k = U$. Such order always exists.

Assume that we are given an order $T$ compatible with $\brd{D}$.  Select and
fix a border $U_i \in \brd{D}$. Let $k = \abs{U_i}$.  We must have $U_{i + 1} =
\maxc{k; T}$.  This implies that if we are given an order $T$ that is
compatible with $\brd{D}$, then $\brd{T} = \brd{D}$.

A naive approach to compute an individual $\maxc{i; T}$ requires $O(NM)$ time,
however, we can compute $\maxc{i; T}$ for each $i$ \emph{simultaneously} in
$O(NM)$ total time using the approach given by~\citet{calders:07:mining}, where
the goal of the authors was to discover the densest interval, essentially
$\maxc{i; T}$, from a stream $T$ in an amortized constant time for each $i$.
This algorithm is actually a variation of PAVA algorithm for solving isotonic
regression for a total order, see~\citep{ayer:55:pav}, though the goal and the
output of the algorithms are different. Connection to isotonic regression is
natural as we saw previously that the exact borders can be discovered by
solving the grid isotonic regression in $O(N^2M^2)$ time.  Since we no longer
consider a grid but a monotonic entry order, we can use a more efficient
algorithm whose copmputational complexity is $O(NM)$.

Let us now consider the optimization problem from an another angle. Assume that
we do not know $\brd{D}$ but instead we know only the average of values in segments,
that is, for each entry $(i, j)$, we know the average, say $w(i, j)$, of the
segment that contains $(i, j)$. We can construct the order compatible with
$\brd{D}$ from the weights. Corollary~\ref{cor:monotone} implies that weights
should decrease. Hence, if we build a monotonic entry order by greedily selecting 
the entry with the largest weight while at the same time making sure that 
the order is indeed monotonic, we end up with an order that is compatible with $\brd{D}$.
We present the pseudo-code for this approach in Algorithm~\ref{alg:findorder}.

\begin{algorithm}
\caption{\findorder, constructs a monotonic order from weights}
\label{alg:findorder}
\Input{weights $w$}
\Output{a monotonic entry order $T$}
$H \define {(1, 1)}$; $T \define$ empty list\;
\While {$H$ is not empty} {
	pop $(i, j)$ from $H$ with the largest weight $w(i, j)$\;
	add $(i, j)$ to $T$\;
	mark $(i, j)$ as visited\;
	\If {$i < M$ \AND ($j = 0$ \OR $(i + 1, j - 1)$ is marked} {
		add $(i + 1, j)$ to $H$\;
	}
	\If {$j < N$ \AND ($i = 0$ \OR $(i - 1, j + 1)$ is marked} {
		add $(i, j + 1)$ to $H$\;
	}
}
\Return $T$\;
\end{algorithm}

We can now construct our algorithm for discovering approximate borders.
Given a monotonic entry order $T = t_1, \ldots, t_{NM}$, we can compute the
borders $\brd{T} = U_0, \ldots, U_L$. Assume that we are given an entry
$p$. Let $U_k \in \brd{T}$ be the border such that $p \in U_k \sm U_{k - 1}$.
Let us define $w_S(p; T) = \freq{U_k \sm U_{k - 1}}$ to be the density
of the segment containing $p$.
Once these weights are computed, we can use them to compute
a new order using \findorder. Computing weights can be done in $O(NM)$ time
while computing a new order can be done in $O(NM \log \min(N, M))$ time.

Using just $w_S$ is problematic in practice. The reason for this is that during
\findorder there will be often several entries in $H$ that belong to
the same segment, and so will have the same values of $w_S$. Hence, we need a
weight function that that will break these ties.  Breaking ties properly is
important since it is possible to design a weight $w$ that $T = \findorder(w(\cdot; T))$ for any $T$. In other words, iterating \findorder 
and computing weights will never improve the current order.

Given an order $T$, we say that a weight function $w$ is a \emph{tie-breaker}
if $w_S(p; T) < w_S(q; T)$ implies $w(p; T) < w(q; T)$, where $p$ and $q$ are
entries in $T$.  Before considering any specific tie-breakers, let us first
show that using a tie-breaker $w$ leads to a convergence.

\begin{proposition}
\label{prop:converge}
Let $w$ be a tie-breaker. Let $T^0$ be any monotonic entry order and define
$T^{i + 1} = \findorder\fpr{w(\cdot; T^i)}$.  Then there exists $k$ such
that $\brd{T^i} = \brd{T^k}$ for any $i \geq k$.
\end{proposition}

\begin{proof}
Fix $m$
and let $U = T^m$ and $V = v_1, \ldots, v_{NM} = \findorder\fpr{w(\cdot; U)}$.
Let $B_0, \ldots, B_L = \brd{U}$ and define $C_0, \ldots, C_K = \brd{V}$.
Define a vector $\alpha$ of length $2L$ such that $\alpha_{2i - 1} = \freq{B_i \sm B_{i - 1}}$
and $\alpha_{2i} = \abs{B_i}$. Define similarly $\beta$ using $\brd{V}$.

Assume two entries $p$ and $q$ such that $p \in B_i$ and $q \notin B_i$.  This
means that $w_S(p; U) > w_S(q; U)$.  Since $w$ is a tie-breaker, this means
that $p$ will occur earlier than $q$ in $V$. In other words, $\{v_1, \ldots, v_{\abs{B_j}}\} = B_j$ as sets, for any $j = 1,\ldots, L$.

If $\brd{U} = \brd{V}$, then it follows immediately that $\alpha = \beta$.

Assume that $\brd{U} \neq \brd{V}$ and
let $j$ be the largest index such that $B_i = C_i$ for any $i \leq j$.
This implies that $\alpha_i = \beta_i$, for $i \leq 2j$.
Since $\{v_{\abs{B_j} + 1}, \ldots, v_{\abs{B_{j + 1}}}\} = B_{j + 1} \sm B_j$ and $C_j = B_j$,
we know that $C_{j + 1} \sm C_j$ will be as dense as $B_{j + 1} \sm B_j$, that is,
it follows that either $\freq{C_{j + 1} \sm C_j} > \freq{B_{j + 1} \sm B_j}$,
or $\freq{C_{j + 1} \sm C_j} = \freq{B_{j + 1} \sm B_j}$ and $\abs{C_{j + 1}} > \abs{B_{j + 1}}$.
That is, either $\alpha_{2j + 1} < \beta_{2j + 1}$, or $\alpha_{2j + 1} = \beta_{2j + 1}$
and $\alpha_{2j + 2} < \beta_{2j + 2}$. Consequently, $\beta$ is larger than $\alpha$ with
respect to the lexicographical order. 

We have shown that if $\brd{T^i} \neq \brd{T^{i + 1}}$, then for any $j >
i$, $\brd{T^i} \neq \brd{T^j}$. Since there are only finite number of possible borders,
there exists an index $k$ such that $\brd{T^i} = \brd{T^k}$ for any $i \geq k$.
\qed
\end{proof}

We consider two tie-breakers. The first tie-breaker, $w_R$ breaks the ties in a
random order. Formally, we define $w_R(p; T) = (w_S(p; T), r)$ to be a vector
of length 2, where $T$ is an order, $p$ is an entry in $T$, and $r$ is a random
number between $0$ and $1$.  The weights are compared in lexicographical order.
This immediately implies that $w_R$ is a tie-breaker. Our second tie-breaker
tries to flip the original order as much as possible. Formally, we define
$w_F(p; T) = (w_S(p; T), i)$, where $T$ is an order, $p$ is an entry in $T$,
and $i$ is the index of $p$ in $T$.  Obviously, $w_F$ is a tie-breaker, and it
will favor the entries that appear later in $T$.

Proposition~\ref{prop:converge} states that the iteration will converge.
However, we need means to detect this convergence. This is difficult since
while the borders themselves will converge, the actual orders may cycle.
Fortunately, we can show that $w_F$ has a cycle of at most 2, we prove this proposition in Appendix~\ref{sec:proofflipcycle}.

\begin{proposition}
Let $T^0$ be any monotonic entry order and define iteratively
$T^{i + 1} = \findorder\fpr{w_F(\cdot; T^i)}$.
There exists $m$ such that $T^{m + 2} = T^m$.
\label{prop:flipcycle}
\end{proposition}

We can easily detect convergence for $w_F$ by remembering the second
last order and compare it to the current one. If the orders are the same,
then we have converged.

Unfortunately, we cannot detect convergence easily for $w_R$. Hence we adopt
the following hybrid approach. We begin with an random monotonic order, and
apply $w_F$ until convergence. Once converged, we apply $w_R$ once, and repeat
applying $w_F$ until we converge again. We repeat this until the borders have not
changed after 20 applications of $w_R$. This is summarized in Algorithm~\ref{alg:heuristic}.

\begin{algorithm}[ht!]
assign all entries to a single border\;
\While {the borders have not changed for $L$ iterations} {
	apply \findorder using random tie-breaker $w_R$\;
	\While {convergence} {
		find borders in ordered entries\;
		apply \findorder using flip tie-breaker $w_F$\;
	}
}
\caption{Heuristic discovery of borders}
\label{alg:heuristic}
\end{algorithm}

\subsection{Using sparsity to speed-up the discovery}\label{sec:sparse}
So far we have operated with ordinary matrices, and consequently the running
time for our heuristic discovery is $O(N^2\log N)$ when applied to a graph with $N$
vertices.  In practice, graphs are sparse and we can use this sparsity to our
advantage.

Assume that we are given a graph $G$ with positive weights and a corner $B$. Let us define
\[
	F = \set{(x, y) \in B \mid (x + 1, y) \notin B \text{ and } (x, y + 1) \notin B}
\]
to be the set of \emph{frontier} points, that is, points in $B$ from which you
cannot advance away from the diagonal, see Figure~\ref{fig:sparse:a}. Now, Proposition~\ref{prop:move} implies
that if $B$ is a border, then $F$ must be a subset of edges. Otherwise, we can
always remove a non-edge frontier point and by doing so increase the density of
$B$. Note also that given $F$, we can always recover $B$ by taking the smallest
corner containing $F$.

This suggests that instead of ordering all entries of the adjacency matrix it is
enough to order just the edges. This, however, poses two complications. First,
we need to be able keep the edges in a monotonic order. While this was easy when
dealing with the full matrix since visiting one entry revealed only 2 adjacent
entries, at maximum, visiting an edge may reveal any number of unvisited edges.
Secondly, in order to compute the density of $B$, whenever we visit an edge we
need to compute the number of non-edges we visited implicitly in order to reach
that edge, see Figure~\ref{fig:sparse:c}. 

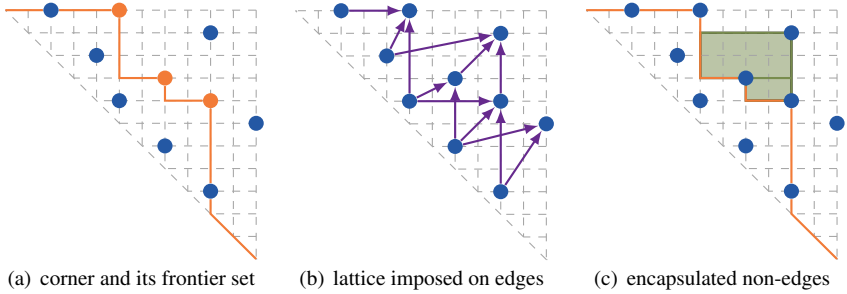
\begin{figure}[htb!]
\tikzstyle{ledge} = [draw, thick, >=latex, ->, yafcolor1]
\tikzstyle{bedge} = [draw, thick, yafcolor2]

\hfill
\subfigure[corner and its frontier set\label{fig:sparse:a}]{
\begin{tikzpicture}[scale = 0.3, xscale = -1]

\drawgrid{0}{11}{dashed, yafaxiscolor!50}

\drawpoints{fill, circle, yafcolor5, inner sep = 2pt}{n}{0/6, 2/3, 2/10, 4/5, 6/7, 7/9, 9/11}
\drawpoints{fill, circle, yafcolor2, inner sep = 2pt}{m}{2/7, 6/11, 4/8}

\draw[bedge] (0, 0) -- (2, 2) -- (n2) -| (m1) -| (m3) -| (m2) -- (n7) -- (11, 11);

\end{tikzpicture}}\hfill
\subfigure[lattice imposed on edges\label{fig:sparse:b}]{
\begin{tikzpicture}[scale = 0.3, xscale = -1]

\drawgrid{0}{11}{dashed, yafaxiscolor!50}

\drawpoints{fill, circle, yafcolor5, inner sep = 2pt}{n}{0/6, 2/3, 2/10, 2/7, 4/5, 6/11, 6/7, 7/9, 9/11, 4/8}

\draw[ledge] (n2) edge (n1);
\draw[ledge] (n5) edge (n1);
\draw[ledge] (n2) edge (n4);
\draw[ledge] (n4) edge (n3);
\draw[ledge] (n5) edge (n4);
\draw[ledge] (n7) edge (n4);
\draw[ledge] (n8) edge (n3);
\draw[ledge] (n9) edge (n6);
\draw[ledge] (n7) edge (n6);
\draw[ledge] (n8) edge (n6);
\draw[ledge] (n10) edge (n3);
\draw[ledge] (n5) edge (n10);
\draw[ledge] (n7) edge (n10);
\end{tikzpicture}}\hfill
\subfigure[encapsulated non-edges\label{fig:sparse:c}]{
\begin{tikzpicture}[scale = 0.3, xscale = -1]

\draw[fill = yafcolor3!50, thick] (2, 7) -| (4, 8) -| (6, 10) -- (2, 10) -- (2, 7);

\drawgrid{0}{11}{dashed, yafaxiscolor!50}

\draw[yafcolor3, thick] (2, 7) -| (4, 8) -| (6, 10) -- (2, 10) -- (2, 7);
\draw[yafcolor3, thick] (4, 8) -- (2, 8);

\drawpoints{fill, circle, yafcolor5, inner sep = 2pt}{n}{0/6, 2/3, 2/10, 4/5, 6/7, 7/9, 9/11}
\drawpoints{fill, circle, yafcolor5, inner sep = 2pt}{m}{2/7, 6/11, 4/8}

\draw[bedge] (0, 0) -- (2, 2) -- (n2) -| (m1) -| (m3) -| (m2) -- (n7) -- (11, 11);
\end{tikzpicture}}
\hspace*{\fill}
\caption{Toy examples related to computing borders from sparse graphs}
\label{fig:sparse}
\end{figure}

In order to solve the first problem, we define a lattice structure on the
\emph{edges}.  More formally, we define a directed acyclic graph $L$ such that
the vertices correspond to the edges of $G$ and the edges in $E(L)$ are formed as
follows: vertex $(x_1, y_1)$ is connected to $(x_2, y_2)$ in $L$ if $x_1 \leq
x_2$ and $y_1 \leq y_2$ and there is no third vertex $(x_3, y_3)$ such that
$x_1 \leq x_3 \leq x_2$ and $y_1 \leq y_3 \leq y_2$, see Figure~\ref{fig:sparse:b}, for example.
In order to guarantee that the edges are visited in a monotonic order, we simply
visit the edges in a topological order of $L$. This modification of \findorder can be done
to run in $O(\abs{E(L)} + \abs{E(G)}\log \abs{V(G)})$ time. Moreover, the lattice can be constructed
in $O(\abs{E(L)}\log \abs{V(G)})$ time.

In order to solve the second problem, that is, to compute the number of
non-edges encapsulated by visiting a new entry, see Figure~\ref{fig:sparse:c},
first note that we need this information when computing the borders from the
monotonic order.  Consequently, at this point we have an order at our disposal.
We enumerate the entries in the order and during this enumeration we compute
and update the frontier set of the border corresponding to the entries visited so
far. This update can be done efficiently by keeping the entries, which are edges in $G$ and have a form, say $(x, y)$, in a red-black tree indexed by $x$.
Adding a new entry into a border will delete at most $k$ entries from the frontier
set, where $k$ is the number of parents in the lattice. There can be at most $\abs{V(G)}$ frontier entries,
at any time. Hence, the running time for this enumeration is $O(\abs{E(L)}\log \abs{V(G)})$. 
Assume now that we have computed the frontier set for $i - 1$ entries, and we need to compute the
encapsulated non-edges for $i$th entry, say $(x, y)$. We do this by finding the entry from the frontier
set $(u, v)$ such that $u$ is the largest possible value and $u \leq x$. 
We then use this entry as a starting point and iterate to the following entries w.r.t. the red-black tree.
At each entry we compute the number of non-edges captured between two adjacent nodes
of the frontier set and $(x, y)$, see Figure~\ref{fig:sparse:c}
for illustration. We need to visit only $O(k)$ entries, where $k$ is the number of parents of $i$th entry
in $L$. Hence, computing the areas for all entries can be done in $O(\abs{E(L)}\log \abs{V(G)})$ time.

\section{Discovering Order}\label{sec:order}

Our main focus so far was to compute the segmentation from already ordered
data. This order may be given, for example, if each vertex has a time stamp.
If the order is not known, we will compute the order using the Fiedler
vector~\citep{fiedler75fiedler}, a popular technique for ordering matrices and graphs.  Fiedler
order has a tendency of pushing the edges towards the diagonal. In fact, we can
show that if there \emph{exists} an order of vertices of the graph $G$ such that
the edges respect the order, that is, all edges are around diagonal,
then Fiedler vector will find this order~\citep{fiedler75fiedler}.

We use Fiedler order as an initial order and introduce a simple heuristic
refinement. Assume that we have computed a segmentation using the order.  Let
$B$ one of the segments and let $(x, y)$ be a frontier, as defined in
Section~\ref{sec:sparse}, of $B$. Find the smallest vertex $x'$ such that
for any $u$, $x' \leq u \leq x$, it holds that for any $v$, $(x, v)$ and $(u, v)$ belong to the
same segment. Similarly, find the largest vertex $y'$ w.r.t. $y$, see Figure~\ref{fig:refine:a}.
It immediately follows that a permutation of vertices between $x'$ and $x$, and $y$ and $y'$
cannot decrease the score since none of the entries will leave its segment.
Assume that there is a non-edge entry $(u, v)$ such that $x' \leq u \leq x$ and $y \leq v \leq y'$.
Then if we swap $x$ and $u$, and $v$ and $y$, we essentially decrease the area
of the segment since none of the entries will leave the segment but at the same
time $(u, v)$ cannot be a frontier since it is non-edge, see Figure~\ref{fig:refine:b}.
If there are several non-edges, we select the non-edge, say $(u, v)$, minimizing
$\abs{(u, z) \in E(G) \mid y \leq z \leq y'} + \abs{(z, v) \in E(G) \mid x' \leq z \leq x}$,
that is, we select $u$ and $v$ with the smallest number of edges.

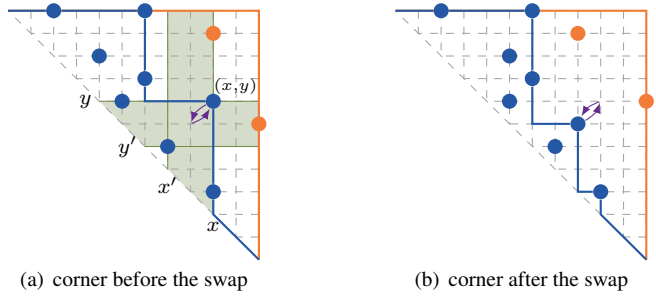
\begin{figure}[htb!]

\hfill
\subfigure[corner before the swap\label{fig:refine:a}]{
\begin{tikzpicture}[scale = 0.3, xscale = -1]

\fill[yafcolor3!30] (2, 2) -- (2, 11) -- (4, 11) -- (4, 4);
\fill[yafcolor3!30] (5, 5) -- (0, 5) -- (0, 7) -- (7, 7);

\drawgrid{0}{11}{dashed, yafaxiscolor!50}

\draw[yafcolor3] (2, 2) -- (2, 11) -- (4, 11) -- (4, 4);
\draw[yafcolor3] (5, 5) -- (0, 5) -- (0, 7) -- (7, 7);

\draw[thick, yafcolor2] (0, 0) -- (0, 11) -- (11, 11);

\drawpoints{fill, circle, yafcolor5, inner sep = 2pt}{n}{2/3, 2/7, 4/5, 5/8, 6/7, 5/11, 7/9, 9/11}
\drawpoints{fill, circle, yafcolor2, inner sep = 2pt}{m}{0/6, 2/10}

\draw[thick, yafcolor5] (0, 0) -- (2, 2) -- (n1) -| (n2) -| (n4) -| (n6) -- (n8) -- (11, 11);

\node[anchor = south west, inner xsep = 0pt] at (n2) {$\scriptstyle (x, y)$};

\node[anchor = north] at (2, 2) {$x$};
\node[anchor = north] at (4, 4) {$x'$};
\node[anchor = east] at (5, 5) {$y'$};
\node[anchor = east] at (7, 7) {$y$};

\node[inner sep = 0pt, outer sep = 0pt] (s) at (3, 6) {};
\draw[->, >=latex, yafcolor1, bend left = 20] (n2) edge (s);
\draw[->, >=latex, yafcolor1, bend left = 20] (s) edge (n2);

\end{tikzpicture}}\hfill
\subfigure[corner after the swap\label{fig:refine:b}]{
\begin{tikzpicture}[scale = 0.3, xscale = -1]

\drawgrid{0}{11}{dashed, yafaxiscolor!50}

\draw[thick, yafcolor2] (0, 0) -- (0, 11) -- (11, 11);

\drawpoints{fill, circle, yafcolor5, inner sep = 2pt}{n}{2/3, 3/6, 4/5, 5/8, 6/7, 5/11, 6/9, 9/11}
\drawpoints{fill, circle, yafcolor2, inner sep = 2pt}{m}{0/7, 3/10}

\draw[thick, yafcolor5] (0, 0) -- (2, 2) -- (n1) -| (n2) -| (n4) -| (n6) -- (n8) -- (11, 11);

\node[inner sep = 0pt, outer sep = 0pt] (s) at (2, 7) {};
\draw[->, >=latex, yafcolor1, bend left = 20] (n2) edge (s);
\draw[->, >=latex, yafcolor1, bend left = 20] (s) edge (n2);

\end{tikzpicture}}
\hspace*{\fill}
\caption{Toy examples related to refining order: before and after swapping $(x, y)$}
\label{fig:refinetoy}
\end{figure}

We will do these swaps until no swaps are possible. For the sake of efficiency,
we do these swaps in a batch style, that is, for each frontier
point $(x, y)$ we first compute $x'$ and $y'$, and then perform swap.
During these swaps we make sure that once an interval $[x', x]$ or $[y, y']$
is used for a swap, it will not be used again in the same batch. Once no swaps
are possible, we recompute the segmentation, and repeat the refinement. This will
eventually converge since the score will always increase after the refinement.

\section{Related Work}\label{sec:related}
To our knowledge our proposed optimization problem is novel.
However, there are similar tasks. A common plotting technique is a contour
plot, where a line follows a constant value. Our optimization technique is
inherently different as we are considering all the points inside a segment and
not just the data points within the vicinity of the boundary.

A related optimization problem was introduced by~\citet{mannila:07:nestedness},
where the authors introduced a concept of nestedness for binary data. A binary data
is nested if for all pairs of rows, one row is either a superset or a subset of
another. If data is fully nested, we can order it in a way that all 1s occur in
the top-left corner. The authors considered minimizing number of 0s that needs
to be replaced such data becomes fully nested. Authors also consider partitioning
columns into $K$ parts such that each part is almost fully segmented.  While
both approaches model the same phenomenon, there are significant differences.
We can handle more general datasets such as counting data and real-valued data,
and our scoring function is based on the log-likelihood. We are interested in
finding $K$-segmentations, whereas \citet{mannila:07:nestedness} focus on
2-segmentation. On the other hand, we assume that our data is preordered
whereas no such assumption is made by~\citet{mannila:07:nestedness}.

The discovered bands give rise for each vertex a set of expanding communities.
Discovering similar structures have been suggested. For example,
\citet{Alvarez-hamelin06largescale} order nodes into a tree by deleting
iteratively low-degree vertices. In another example, \citet{tatti:13:nested}
suggested discovering nested communities given a set of seed nodes.

The spectral approach to discover an order falls into a larger category of approaches
used for solving the seriation problem.
In seriation, a typical goal is, given a similarity matrix between objects, to
organize the objects that minimize some stress function~\citep{hahler08order}. In the ideal case, the
discovered order will rearrange the similarity matrix into a \emph{Robinson}
form, that is, the values of the matrix will get smaller as we move away from
the diagonal. In practice, such a permutation rarely exists and hence the
stress function reflects how far away we are from the Robinson form.
\citet{fiedler75fiedler} demonstrated that the spectral method, that is, ordering
using the Fiedler vector, will find the Robinson form \emph{if} such exists.
If such an order does not exist, then the spectral order can be viewed as an
algorithm for placing objects on a straight line with the goal of minimizing
pair-wise distance weighted by the similarity matrix. Note that the goal here
is to find the locations on a line, the order of the objects comes as a
by-product. Optimizing stress functions that are directly based on the order of
the similarity matrix is typically a \textbf{NP}-hard problem. In such cases
the problem typically resembles a travelling salesman problem. Consequently,
heuristics to discover a good order are variants of TSP solvers~\citep{conf/vldb/JohnsonKCKV04}.

Ordering data rows and columns and using this order to discover an underlying
structure has been suggested.
\citet{gionis:04:geometric}~and~\citet{tatti:12:discovering} suggested
discovering tile hierarchies from ordered binary datasets. Both work used
Bernoulli models, a log-linear model, to score the hierarchies.

In our experiments, the number of borders is small. Hence, we are able to solve
Problem~\ref{prob:seqseg} exactly using a dynamic program. This program
requires quadratic time. If the number of borders becomes too large, that is,
close to the number of entries, using an exact solver becomes impractical.  In
order to overcome this problem we can use heuristic approaches, such as
top-down approaches where we select greedily a new change-point
(see~\citet{shatkay:96:approximate,bernaola-galvan:96:compositional,douglas:73:algorithms,lavrenko:00:mining},
for example) or bottom-up approaches where at the beginning each point is a
segment, and points are combined in a greedy fashion
(see~\citet{palpanas:04:online}, for example). Yet another option is a
randomized heuristic was suggested by~\citet{himberg:01:time}, where we start
from a random segmentation and optimize the segment boundaries.  These
approaches, although fast, are heuristics and have no theoretical guarantees of
the approximation quality.  A divide-and-segment approach, an approximation
algorithm with theoretical guarantees on the approximation quality was given
by~\citet{terzi:06:efficient}.

\section{Experiments}\label{sec:experiments}

\paragraph{Setup:}
In our experiments we used 6 real-world datasets and one synthetic data. The first two graphs,
\emph{DblpCP} and \emph{DblpCF}, are ego-networks of Christos Papadimitriou and
Christos Faloutsos, that is, the graphs contain the collaborators obtained from
DBLP,\!\footnote{\url{http://www.informatik.uni-trier.de/~ley/db/}} two
researchers are connected if they have a joint paper. The other two datasets,
\emph{Fb107} and \emph{Fb1912}, were the two largest ego-networks taken from
the Facebook dataset obtained from SNAP
repository.\!\footnote{\url{http://snap.stanford.edu/snap/}}
The \emph{Mammals} presence data consists of presence records of European
mammals within geographical areas of $50\times50$
kilometers~\citep{mitchell-jones:99:atlas}.\!\footnote{The full version of the
\emph{Mammals} dataset is available for research purposes from the Societas
Europaea Mammalogica at \url{http://www.european-mammals.org}} We computed a
similarity matrix between different locations, say $i$ and $j$, by considering
the observed joint number of different mammals, normalized by $c_ic_j$, where
$c_i$ and $c_j$ are the number of mammals observed at $i$ and $j$,
respectively. Dataset \emph{Paleo}\footnote{NOW
public release 030717 available from~\cite{fortelius05now}.}  contains
information of species fossils found in specific paleontological sites in
Europe~\cite{fortelius05now}. We constructed a similarity matrix between two
fossils, say $a$ and $b$, by computing the number of sites in which both $a$
and $b$ have been discovered.
We also created a synthetic dataset with $250\,000$
vertices and $279\,223$ edges where edges had a higher probability of being
closer to the diagonal. The basic characteristics of the datasets are given in
Table~\ref{tab:basic}.

We ordered the vertices of each graph (except synthetic) according to Fiedler's vector using
algorithm given by~\citet{atkins99seriation}. Since Atkins' algorithm may
produce several orders, this happens when the graph has disconnected
components, we pick one order at random.  Once the graph is ordered we apply
our heuristic method given in Section~\ref{sec:heuristic} along with the
refinement step described in Section~\ref{sec:order}.  We also applied the
algorithm for discovering for exact borders described in Section~\ref{sec:exact}.
We used Bernoulli model for unweighted graphs, Poisson model for \emph{Paleo} data
and $L_2$ error for \emph{Mammals}.
We set the number of bands to $3$ in DBLP and \emph{Paleo} datasets and $4$ in Facebook and \emph{Mammals} datasets.
The statistics of the datasets and the experiments are given in Table~\ref{tab:basic}.
Due to the slow convergence of heuristic method in \emph{Mammals} and \emph{Synthetic},
we limited the number of iterations to $2000$ and did not apply refinement.
We also noticed that after each random tie-breaker the score gets a (relatively
small) bump. Hence, we force random tie-breaker after each $50$th iteration in these two datasets.

\begin{table}[htb!]
\caption{Basic characteristics of the datasets, 
the number of edges in lattice, see Section~\ref{sec:sparse}.}
\label{tab:basic}
\begin{tabular*}{\textwidth}{@{\extracolsep{\fill}}l*{3}{r}}
\toprule
Name & $\abs{V}$ & $\abs{E}$ & $E(L)$ \\
\midrule
DblpCF & 176  & 457     & 794    \\
DblpCP & 154  & 348     & 717     \\
Fb107  & 1034 & 26\,749 & 61\,279 \\
Fb1912 & 747  & 30\,025 & 62\,842\\
Paleo  & 139  & 4428    & 8737 \\
Mammals& 2183 & 2\,378\,193& 4\,752\,044\\
Synthetic & 250\,000 & 279\,223 & 2\,124\,683 \\
\bottomrule
\end{tabular*}
\end{table}

\begin{table}[htb!]
\caption{Basic statistics from experiments using exact and approximate borders.
The columns contain running times, the number
of borders, the number of flip tie-breakers, the number of random tie-breakers,
the number of refinement rounds, initial (negative) score based on a fiedler order, and final score after the refinement.}
\label{tab:runs}
\begin{tabular*}{\textwidth}{@{\extracolsep{\fill}}l*{7}{r}@{}r*{3}r}
\toprule
& \multicolumn{7}{l}{Approximate borders} && \multicolumn{3}{l}{Exact borders} \\
\cmidrule{2-8}
\cmidrule{10-12}
Name & time & brd & iter & rnd & ref & initial & final && time & initial & final \\
\midrule
DblpCF & 0.2s & 55  & 235  & 48  & 3 & 945     & 908     && .02s & 945 & 905 \\
DblpCP & 0.4s & 53  & 701  & 135 & 3 & 966     & 927     && .05s & 966 & 918\\
Fb107  & 12m  & 476 & 7217 & 676 & 7 & 61\,734 & 60\,444 && 20s   & 61\,723& 60\,427\\
Fb1912 & 5m   & 375 & 7357 & 813 & 4 & 43\,212 & 42\,909 && 3.2s  & 43\,212 & 42\,930  \\
Paleo  & 4s   & 201 & 423  & 51  & 4 & $-8645$ & $-8906$ && .13s & $-8645$ & $-8906$   \\
Mammals& 33m  & 2975& 2000 & 40  &   & 19\,798 &         && 2m & 19\,798 & 19\,798 \\
Synthetic & 37m& 625& 2000 & 40  &   & 6\,956\,048\\
\bottomrule
\end{tabular*}
\end{table}

\paragraph{Results:}
Let us first focus on computational complexity. We see that the exact algorithm
works very well in practice, it is faster than the heuristic discovery, and it
has the benefit of producing the exact result.  The main reason for this is
that for medium-sized datasets the running time for the exact algorithm is very
competitive and the slowness of the heuristic algorithm is due to high number
of iterations it requires to converge. 
This dynamic changes when the size of
the data increases. For large datasets, the fact that we need to store the
graph in the full form will slow the algorithm down (or simply run out of
memory). This is the case with the \emph{Synthetic} data where we could not run
the exact algorithm due to the memory limitations.

Unlike the exact method, 
the heuristic discovery of bands can use sparsity to its
advantage (see Section~\ref{sec:sparse}).
More specifically, one iteration is linear w.r.t. the number of
edges in the lattice constructed from the edges of the original graph.  While
it is possible in theory that the number of lattice edges is significantly
higher than the number of edges in the original graph, third column in
Table~\ref{tab:basic} demonstrates that in our experiments the number of
lattice edges is about 2--3 times the number of edges in the original graph. Due
to this property the running times stay reasonably small, minutes at worst.

The number of discovered borders, 5th column, is small, which in turns implies
that the final segmentation is fast, and indeed we spend most of the time
computing the borders.

Let us now study convergence of the heuristic in more detail.
Table~\ref{tab:runs} shows that the heuristic converges close to the exact
values w.r.t. the score function.
The number of
iterations needed in a single refinement round, given in $3$th and $4$th
columns of Table~\ref{tab:basic}, grows as the graph gets larger. The number of
random tie-breaker iterations is about $10$ percent. In order to study the
convergence in more details, we plotted $\score{\sgm{B}}$ in
Figure~\ref{fig:convergence}, where $\sgm{B}$ is the current set of approximate
borders, as a function of rounds left to convergence.  Since the initial border
set comes from a random monotonic order, we repeated this experiment 10 times.
Note that $\score{\sgm{B}}$ is an upper bound for score of the final
segmentation. From the results we see that initial orders have a bad score, and
there is a high variance in the initial score and the convergence time.
However, the score stabilizes quickly, close
to its final value, and the final score has little dependence of the starting
point. Majority of the rounds is spent in cosmetic improvement. This hints that
if the convergence time is a factor, we may stop the border discovery early and
still get a good score.  This is demonstrated with the \emph{Synthetic} dataset.
The convergence of the score is given in Figure~\ref{fig:convergence} and again
we see the same behaviour where the initial order has a bad score but quickly
obtains a stable score.  Running the heuristic took $37$ minutes and we were not
able to run the exact algorithm due to the memory limitations.

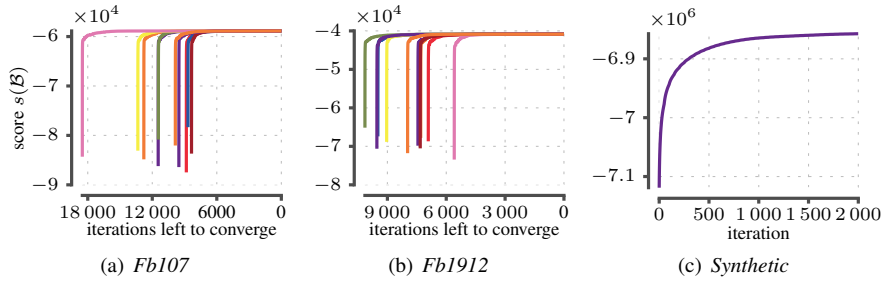
\begin{figure}[htb!]
\hfill
\subfigure[\emph{Fb107}]{
\begin{tikzpicture}
\begin{axis}[xlabel=iterations left to converge, ylabel= {score $\score{\sgm{B}}$},
    width = 4.2cm,
	tick scale binop = \times,
    xtick = {-18000, -12000, -6000, 0},
    xticklabels = {$18\,000$, $12\,000$, 6000, 0},
    x tick label style = {/pgf/number format/set thousands separator = {\,}},
	scaled x ticks = false,
    ytick = {-60000, -70000, -80000, -90000},
	ymin = -90000,
    cycle list name=yaf,
	no markers,
    ]
\addplot table[x expr = {-\thisrowno{0}}, y index = 1, header = false] {data/fb107_1.log};
\addplot table[x expr = {-\thisrowno{0}}, y index = 1, header = false] {data/fb107_2.log};
\addplot table[x expr = {-\thisrowno{0}}, y index = 1, header = false] {data/fb107_3.log};
\addplot table[x expr = {-\thisrowno{0}}, y index = 1, header = false] {data/fb107_4.log};
\addplot table[x expr = {-\thisrowno{0}}, y index = 1, header = false] {data/fb107_5.log};
\addplot table[x expr = {-\thisrowno{0}}, y index = 1, header = false] {data/fb107_6.log};
\addplot table[x expr = {-\thisrowno{0}}, y index = 1, header = false] {data/fb107_7.log};
\addplot table[x expr = {-\thisrowno{0}}, y index = 1, header = false] {data/fb107_8.log};
\addplot table[x expr = {-\thisrowno{0}}, y index = 1, header = false] {data/fb107_9.log};
\addplot table[x expr = {-\thisrowno{0}}, y index = 1, header = false] {data/fb107_10.log};
\pgfplotsextra{\yafdrawaxis{-18505}{0}{-90000}{-58865}}
\end{axis}
\end{tikzpicture}}
\hfill
\subfigure[\emph{Fb1912}]{
\begin{tikzpicture}
\begin{axis}[xlabel=iterations left to converge, ylabel= {}, 
    width = 4.2cm,
	tick scale binop = \times,
    xtick = {-9000, -6000, -3000, 0},
    xticklabels = {$9\,000$, $6\,000$, $3\,000$, 0},
    x tick label style = {/pgf/number format/set thousands separator = {\,}},
    ytick = {-40000, -50000, -60000, -70000, -80000},
	ymin = -80000,
	ymax = -40000,
    cycle list name=yaf,
	no markers,
	scaled x ticks = false,
    ]
\addplot table[x expr = {-\thisrowno{0}}, y index = 1, header = false] {data/fb1912_1.log};
\addplot table[x expr = {-\thisrowno{0}}, y index = 1, header = false] {data/fb1912_2.log};
\addplot table[x expr = {-\thisrowno{0}}, y index = 1, header = false] {data/fb1912_3.log};
\addplot table[x expr = {-\thisrowno{0}}, y index = 1, header = false] {data/fb1912_4.log};
\addplot table[x expr = {-\thisrowno{0}}, y index = 1, header = false] {data/fb1912_5.log};
\addplot table[x expr = {-\thisrowno{0}}, y index = 1, header = false] {data/fb1912_6.log};
\addplot table[x expr = {-\thisrowno{0}}, y index = 1, header = false] {data/fb1912_7.log};
\addplot table[x expr = {-\thisrowno{0}}, y index = 1, header = false] {data/fb1912_8.log};
\addplot table[x expr = {-\thisrowno{0}}, y index = 1, header = false] {data/fb1912_9.log};
\addplot table[x expr = {-\thisrowno{0}}, y index = 1, header = false] {data/fb1912_10.log};
\pgfplotsextra{\yafdrawaxis{-10145}{0}{-80000}{-40000}}
\end{axis}
\end{tikzpicture}}\hfill
\subfigure[\emph{Synthetic}]{
\begin{tikzpicture}
\begin{axis}[xlabel=iteration, ylabel= {}, 
    width = 4.2cm,
	tick scale binop = \times,
    xtick = {-2000, -1500, -1000, -500, 0},
    xticklabels = {$0$, $500$, $1\,000$, $1\,500$, $2\,000$},
    x tick label style = {/pgf/number format/set thousands separator = {\,}},
    cycle list name=yaf,
	no markers,
    ]
\addplot table[x expr = {-\thisrowno{0}}, y index = 1, header = false] {data/large.log};
\pgfplotsextra{\yafdrawaxis{-2000}{0}{-7118949}{-6857371}}
\end{axis}
\end{tikzpicture}}
\hspace*{\fill}
\caption{Convergence during heuristic discovery. Score $\score{\sgm{B}}$, where $\sgm{B}$ is the current set of borders
as a function of iterations left to convergence. Note that $\score{\sgm{B}}$ is an upper bound for the score of final segmentation.}
\label{fig:convergence}
\end{figure}

Let us next look at the refinement. In our experiments, the number of
refinement iterations is low, around 10.  Moreover, the improvement in score is
modest, around $1$\%. This is also the case when we apply to the datasets that are not
ordered with the fiedler vector.  This suggests that the refinement procedure
should not be used alone but instead it should be used only to improve the
already good order. We also tried discovering bands using random orders.
Naturally, the score for discovered bands is weaker. The gain of using spectral
orders depends how strongly banded is the data. For example, \emph{Paleo}
contains a strong banded structure, and the score using a random order was
significantly worse than the spectral order.
Finally, using the exact border discovery with the refinement may produce a worse
score than using the approximate border discovery, as shown by the scores for \emph{Fb1912}.

Finally, let us look at the discovered bands given in
Figures~\ref{fig:wbands}--\ref{fig:snipet}.  The first band in the DBLP datasets
is relatively small, typically discovering small clusters that are due to joint
publication with the author of the ego-network. Facebook graphs contain a
strong clustering structure which is seen in the discovered bands. However, we
also see the overlap effect: in \emph{Fb1912} the clusters overlap.  In
Figure~\ref{fig:snipet} we provide small snipets from DBLP graphs along with
the authors. For example, in \emph{DblpCP}, the inner band contains authors that have collaborate
with each other but these authors do not form a clique. In \emph{Paleo}, the bands correspond
to the weights of the edges: the inner band has edges with higher weights while the outer bands
have less edges.  In \emph{Mammals}, the spectral order---which was computed
without the knowledge of the geological location---corresponds roughly to the latitude: one extreme being
Scandinavia while on the other extreme being Greece and Crete. The order fails to capture Spain and Portugal
suggesting that the data cannot be explained fully by a single order. The discovered bands, especially in the northern
areas, suggest a continuous change in fauna as we move along the latitude.

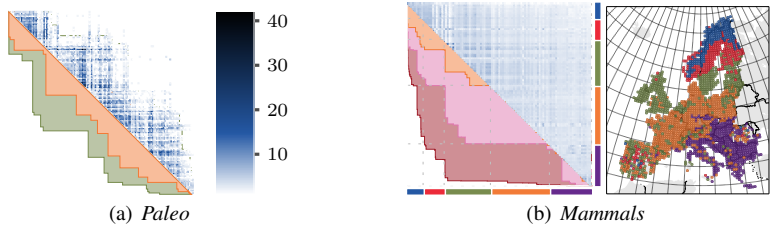
\begin{figure}[htb!]
\hfill
\subfigure[\emph{Paleo}]{
\begin{tikzpicture}
\begin{axis}[width = 4cm,height = 4cm, ticks = none, grid=none,
	every axis plot post/.style={},
	colormap name = cool,
	colorbar,
	colorbar style = {axis y line*=right, y tick style= {xshift=-\yafaxispad}, y tick label style = {xshift = -\yaftlpad}}
	]

\drawgraphvalue{data/paleo-dist.out}{data/paleo-dist.ord}

\end{axis}
\end{tikzpicture}}\hspace*{0.5cm}\hfill
\subfigure[\emph{Mammals}\label{fig:mammals}]{
\begin{tikzpicture}
\begin{axis}[width = 4cm,height = 4cm, ticks = none, 
	scatter/classes ={0={yafcolor5},1={yafcolor2},2={yafcolor3},3={yafcolor4}},
	every axis plot post/.style={}, 
	axis on top,
	xtick = {0, -500, -1200, -1750, -2000},
	ytick = {0, 500, 1200, 1750, 2000}
	]
\drawvalues{data/mammals-dist-sparse.ord}{cooler}
\drawvarea{data/mammals-dist.out}{2}{yafcolor7}
\drawvarea{data/mammals-dist.out}{1}{yafcolor8}
\drawvarea{data/mammals-dist.out}{0}{yafcolor2}
\draw[yafcolor1, fill] (axis cs: 50, 1) rectangle (axis cs: 100, 470);
\draw[yafcolor2, fill] (axis cs: 50, 501) rectangle (axis cs: 100, 1170);
\draw[yafcolor3, fill] (axis cs: 50, 1201) rectangle (axis cs: 100, 1720);
\draw[yafcolor4, fill] (axis cs: 50, 1751) rectangle (axis cs: 100, 1970);
\draw[yafcolor5, fill] (axis cs: 50, 2001) rectangle (axis cs: 100, 2182);

\draw[yafcolor1, fill] (axis cs: -1, -50) rectangle (axis cs: -470, -100);
\draw[yafcolor2, fill] (axis cs: -501, -50) rectangle (axis cs: -1170, -100);
\draw[yafcolor3, fill] (axis cs: -1201, -50) rectangle (axis cs: -1720, -100);
\draw[yafcolor4, fill] (axis cs: -1751, -50) rectangle (axis cs: -1970, -100);
\draw[yafcolor5, fill] (axis cs: -2001, -50) rectangle (axis cs: -2182, -100);
\end{axis}
\end{tikzpicture}
\includegraphics[height=2.5cm]{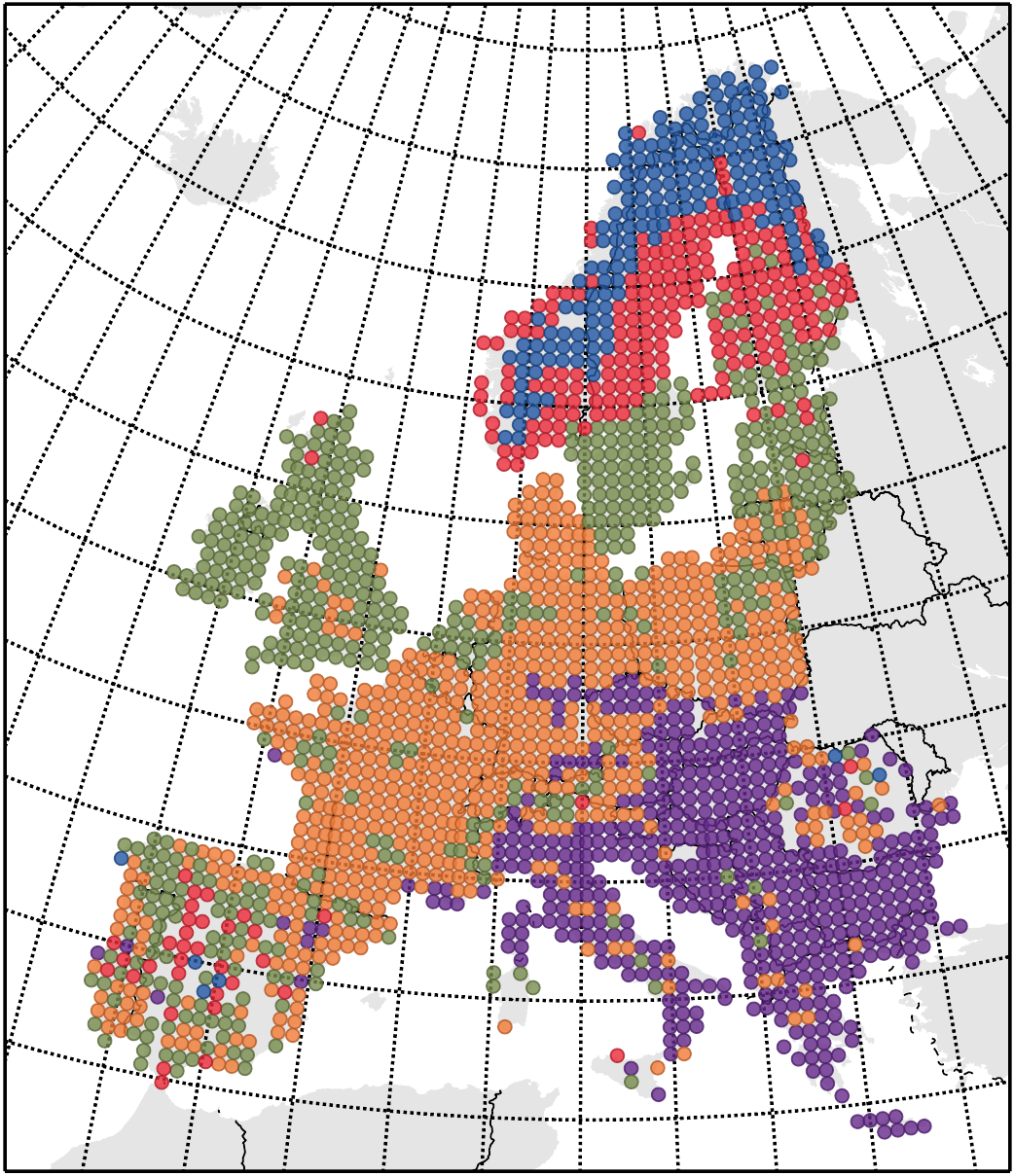}}\hspace*{\fill}
\caption{Adjacency matrices and the obtained segmentations. The upper triangle
depicts the weights of edges while the lower triangle shows the
segmentation. \emph{Mammals} was sparsified for visualization purposes.
In~\ref{fig:mammals}, the colors of the bars along the axis correspond to the color of the locations in the map.}
\label{fig:wbands}
\end{figure}

\begin{figure}[htb!]
\subfigure[\emph{DblpCF}]{
\begin{tikzpicture}
\begin{axis}[width = 4cm,height = 4cm, ticks = none, grid=none,
	scatter/classes ={0={yafcolor5},1={yafcolor2},2={yafcolor3},3={yafcolor4}},
	every axis plot post/.style={}]

\drawgraph{data/ego-christos-faloutsos-ex.out}{data/ego-christos-faloutsos-ex.ord}

\end{axis}
\end{tikzpicture}}\hfill
\subfigure[\emph{DblpCP}]{
\begin{tikzpicture}
\begin{axis}[width = 4cm,height = 4cm, ticks = none, grid=none,
	scatter/classes ={0={yafcolor5},1={yafcolor2},2={yafcolor3},3={yafcolor4}},
	every axis plot post/.style={}]

\drawgraph{data/ego-christos-papadimitriou-ex.out}{data/ego-christos-papadimitriou-ex.ord}

\end{axis}
\end{tikzpicture}}\hfill
\subfigure[\emph{Fb107}]{
\begin{tikzpicture}
\begin{axis}[width = 4cm,height = 4cm, ticks = none, grid=none,
	scatter/classes ={0={yafcolor5},1={yafcolor2},2={yafcolor3},3={yafcolor4}},
	every axis plot post/.style={}]

\drawgraphiv{data/facebook-107-ex.out}{data/facebook-107-sparse.ord}

\end{axis}
\end{tikzpicture}}\hfill
\subfigure[\emph{Fb1912}]{
\begin{tikzpicture}
\begin{axis}[width = 4cm,height = 4cm, ticks = none, grid=none,
	scatter/classes ={0={yafcolor5},1={yafcolor2},2={yafcolor3},3={yafcolor4}},
	every axis plot post/.style={}]

\drawgraphiv{data/facebook-1912-ex.out}{data/facebook-1912-sparse.ord}

\end{axis}
\end{tikzpicture}}
\caption{Adjacency matrices and the obtained segmentations. The upper triangle
depicts the distribution of edges while the lower triangle shows only the
segmentation. The facebook graphs were sparsified for visualization purposes}
\label{fig:bands}
\end{figure}
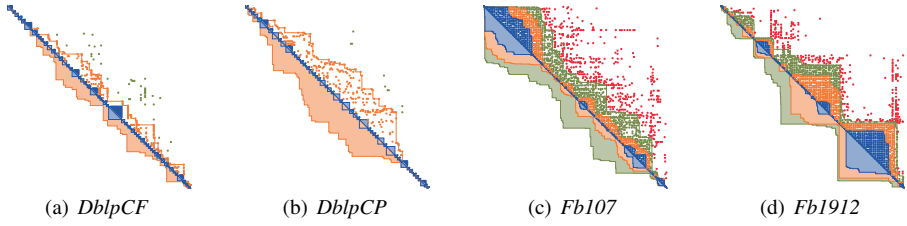

\begin{figure}[htb!]
\hfill
\subfigure[close-up of \emph{DblpCP}]{
\begin{tikzpicture}
\begin{axis}[ width = 4cm,height = 4cm,
	ticks = none,
	scatter/classes ={0={yafcolor5, mark = square*},1={yafcolor2},2={yafcolor3, mark = triangle*},3={yafcolor4}},
	every axis plot post/.style={},
	xtick = {-128,-130,...,-139},
	ytick = {128,130,...,139},
	mark size = 2pt]

\drawdata{data/ego-christos-papadimitriou-project.ord}
\addplot[nodes near coords,
	every node near coord/.append style={anchor=base west, font = \tiny, inner ysep = 0pt},
	only marks, point meta = explicit symbolic]
	table[col sep = tab, x expr = {-128}, y expr = {\thisrowno{0}  - 0.3}, meta index = 1, header = false] {data/map-christos-papadimitriou-project.txt};
\addplot[nodes near coords,
	every node near coord/.append style={anchor=base west, font = \tiny, inner sep = 0pt, rotate = -45},
	only marks, point meta = explicit symbolic]
	table[col sep = tab, y expr = {127}, x expr = {-(\thisrowno{0} + 0.5)}, meta index = 1, header = false] {data/map-christos-papadimitriou-project-abbrv.txt};

\end{axis}
\end{tikzpicture}}\hfill
\subfigure[close-up of \emph{DblpCF}]{
\begin{tikzpicture}
\begin{axis}[ width = 4cm,height = 4cm,
	ticks = none,
	scatter/classes ={0={yafcolor5, mark = square*},1={yafcolor2},2={yafcolor3, mark = triangle*},3={yafcolor4}},
	every axis plot post/.style={},
	xtick = {-29,-31,...,-40},
	ytick = {29,31,...,40},
	mark size = 2pt]

\drawdata{data/ego-christos-faloutsos-ex-project.ord}
\addplot[nodes near coords,
	every node near coord/.append style={anchor=base west, font = \tiny, inner ysep = 0pt},
	only marks, point meta = explicit symbolic]
	table[col sep = tab, x expr = {-29}, y expr = {\thisrowno{0}  - 0.3}, meta index = 1, header = false] {data/map-christos-faloutsos-project.txt};
\addplot[nodes near coords,
	every node near coord/.append style={anchor=base west, font = \tiny, inner sep = 0pt, rotate = -45},
	only marks, point meta = explicit symbolic]
	table[col sep = tab, y expr = {28}, x expr = {-(\thisrowno{0} + 0.5)}, meta index = 1, header = false] {data/map-christos-faloutsos-project-abbrv.txt};

\end{axis}
\end{tikzpicture}}\hspace*{\fill}
\caption{Snipets of \emph{DblpCP} and \emph{DblpCF}}
\label{fig:snipet}
\end{figure}
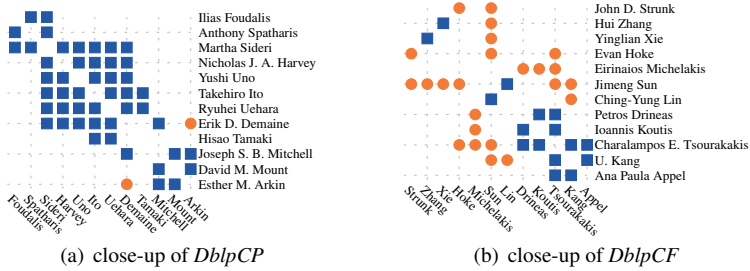

\section{Conclusion}\label{sec:conclusion}
In this paper we introduced an optimization problem in order to model the
concentration of edges next to the diagonal. More specifically, given a graph
and an integer $K$, our goal is to order the vertices and find $K$ bands
around the diagonal, inner bands being more dense. 
As a measure of goodness we use the likelihood of a log-linear model, a
family containing many standard distributions.

We divide the problem into two subproblems. The first problem is to find a good
order while the second problem is to discover the bands.  We introduced two
solvers for the latter problem. The first solver is exact and based on isotonic
regression, the second solver is a heuristic approach that sorts the entries
into a linear order and applies one-dimensional isotonic regression. By doing
so, the solver can exploit sparsity of the graph.  Both approaches complement
each other. If we can store the graph as a full adjacency matrix, then it is
beneficial to use the exact algorithm. The algorithm runs in $O(N^4)$ time, at
worst, but is closer to $O(N^2\log P)$, where $N$ is the number of vertices and
$P$ is the number of borders.  On the other hand, if the number of vertices is
large but the graph is sparse it is better to apply the heuristic approach with
a limited number of iterations. A single iteration runs in $O(\abs{E(L)})$
time, where $L$ is the lattice describing the relations between the nodes
(described in Section~\ref{sec:sparse}). At worst, $\abs{E(L)}$
is $O(N^2)$ but in practice it is smaller for sparse graphs.

We discover the order mainly by using Fiedler order.  It would be fruitful to
develop algorithms that induce an order while looking for the optimal
segmentation simultaneously. The work done by~\citet{mannila:07:nestedness}
with nested datasets hints that this problem is \textbf{NP}-hard. However, it
may be possible to develop an efficient heuristic approach.

Another open question, which we will leave to a future work, is the problem of
choosing $K$, the number of segments.  This question is not only specific to
our problem but also occurs in many classic problems such as clustering or
sequence segmentation. In our experiments computing the final segmentation from
the set of borders is cheap.  Hence, instead of studying a single segmentation
we can compute several segmentations simultaneously and present all of them to the user.
Alternatively, since our score is essentially a log-likelihood, we can design an
MDL approach to select a segmentation with a good score.

\bibliographystyle{abbrvnat}
\bibliography{abbrev,bibliography}

\ifapx
\appendix
\section{Proof of Proposition~\lowercase{\ref{prop:borders}}}\label{sec:proofborder}

In order to prove this proposition, we first need to establish a series of
lemmas.  We will use the following fact. Let $B$ be a set of entries, and
let $C \subsetneq B$ be a non-empty subset of $B$ such that $\freq{C} <
\freq{B}$, then $\freq{B \sm C} > \freq{B}$.

The first lemma states that if you cut a segment from a maximal corner, then that
segment will be more dense. Conversely, any segment adjacent to a maximal corner
will be more sparse.

\begin{lemma}
\label{lem:split}
Let $U$ be a corner. Let $V = \maxc{U}$ be a corner.
Let $X \supseteq U$ be a corner. If $V \sm X \neq \emptyset$,
then $\freq{V \sm X} \geq \freq{V \sm U}$.
If $X \sm V \neq \emptyset$,
then $\freq{X \sm V} < \freq{V \sm U}$.
\end{lemma}

\begin{proof}
Assume that $\freq{V \sm X} < \freq{V \sm U}$.
Then, since
\[
	(V \sm U) \sm (V \sm X) = (V \cap X) \sm U,
\]
$\freq{(V \cap X) \sm U} > \freq{V \sm U}$,
which contradicts the definition of $V$.
The second case holds because otherwise $\freq{(V \cup X) \sm U} \geq \freq{V \sm U}$,
which contradicts the definition of $V$.
\qed\end{proof}

The next lemma essentially states that if a corner cuts a maximal corner, then
it cannot be a border.

\begin{lemma}
\label{lem:augment}
Let $U$ be a corner. Write $V = \maxc{U}$.
Let $X \supsetneq U$ be a corner such that $V \sm X \neq \emptyset$, then
$\freq{(V \cup X) \sm X} \geq \freq{X \sm U}$.
Consequently, $X$ is not a border.
\end{lemma}

\begin{proof}
By definition of $V$, $\freq{V \sm U} \geq \freq{X \sm U}$.

Lemma~\ref{lem:split} implies that $\freq{V \sm X} \geq \freq{V \sm
U}$. As $(V \cup X) \sm X = V \sm X$, we can combine the two inequalities and prove the first claim.
$X$ cannot be a border since $U \subsetneq X \subsetneq V \cup X$.
\qed\end{proof}

Lemmas~\ref{lem:split}--\ref{lem:augment} can be modified to hold for minimal
corners. The proofs for these lemmas are similar to the proofs of Lemmas~\ref{lem:split}--\ref{lem:augment}.

\begin{lemma}
\label{lem:split2}
Let $U$ be a corner. Let $V = \minc{U}$ be a corner.
Let $X \subsetneq U$ be a corner. If $X \sm V \neq \emptyset$,
then $\freq{X \sm V} \leq \freq{U \sm V}$.
If $V \sm X \neq \emptyset$,
then $\freq{V \sm X} > \freq{U \sm V}$.
\end{lemma}

\begin{lemma}
\label{lem:augment2}
Let $U$ be a corner. Write $V = \minc{U}$.
Let $X \subsetneq U$ be a corner such that $X \sm V \neq \emptyset$, then
$\freq{X \sm (V \cap X)} \leq \freq{U \sm X}$.
Consequently, $X$ is not a border.
\end{lemma}

Assume that we are given two $K$-segmentations, $\sgm{U} = U_0, \ldots, U_K$
and $\sgm{V} = V_0, \ldots, V_K$. We write $\sgm{U} \precneqq \sgm{V}$ if there
exists $i$ such that $U_j = V_j$ for $j < i$, and $U_i \subsetneq V_i$.  The
idea behind the proof of Proposition~\ref{prop:borders} is that we can replace
a non-border in a segmentation, say $\sgm{U}$, either with its minimal or
maximal corner such that either the score will increase or that the resulting
segmentation, say $\sgm{V}$, is smaller wrt. partial order, $\sgm{V} \precneqq
\sgm{U}$.

Next lemma makes sure that we can always pick a non-border in a
segmentation that such that its maximal and minimal corners form a chain with
the other corners in the segmentation.

\begin{lemma}
\label{lem:squeeze}
Let $U_0, \ldots, U_K$ be a $K$-segmentation.  Assume that, say, $U_i$ is not a
border. Then there exists $U_j$ with $0 < j < K$, and two corners $X$ and $Y$ such that
$U_{j - 1} \subseteq X \subsetneq U_j \subsetneq Y \subseteq U_{j + 1}$ and $\freq{Y
\sm U_j} \geq \freq{U_j \sm X}$.
\end{lemma}

\begin{proof}
Since $U_i$ is not a border, there exist $X$ and $Y$ such that $X \subsetneq U_i
\subsetneq Y$ and $\freq{Y \sm U_i} \geq \freq{U_i \sm X}$.  We need to
show that $U_{i - 1} \subseteq X$ and $Y \subseteq U_{i + 1}$, or possibly
modify $X$, $Y$, and $i$ such that inclusions hold.

We can safely assume that $Y = \maxc{U_i}$.

Assume that $Y \nsubseteq U_{i + 1}$, that is, $Y \sm U_{i + 1} \neq
\emptyset$.  Lemma~\ref{lem:augment} implies that $\freq{Y \sm U_{i + 1}}
\geq \freq{U_{i + 1} \sm U_i}$. Hence, we can redefine $X = U_i$ and $Y =
\maxc{U_{i + 1}}$, and then increase $i$ by one and repeat the argument.  This
process will eventually stop since we increase $i$ every step.  When we finally
stop, note that both inclusions will hold, and we have proved the lemma.

If we start with the case where $Y \subseteq U_i$ and $U_{i - 1} \nsubseteq X$,
we proceed to modify $X$, $Y$, and $i$ in the opposite direction.
\qed\end{proof}

Our next goal is Corollary~\ref{cor:monotonepar} which will make sure that
we can always make sure that the segmentation is monotonic. For that we need
the following two lemmas, that follow immediately from the fact that our
scoring function is a log-linear model.

\begin{lemma}
\label{lem:concave}
$\score{X \mid u}$ is a concave function of $u$. 
Let $U$ and $V$ be two segments such that $\freq{U} \leq \freq{V}$.
Let $r = \arg \max_u \score{U \mid u}$ and $t = \arg \max_u \score{V \mid u}$.
Then $r \leq t$.
\end{lemma}

\begin{proof}
Let $X$ be a segment.
A straightforward calculation shows that $\partial \score{X \mid u} / \partial
u = \abs{X}(\freq{X} - \mean{u}{x})$, where $\mean{u}{x} = \sum_x x p(x \mid u)$ is the
mean of the log-linear model.

We also have $\partial \mean{u}{x} / \partial u = \var{u}{x} \geq 0$. This
immediately proves that $\score{X \mid u}$ is a concave function.  The optimal
(possibly infinite) value is reached when $\freq{X} = \mean{u}{x}$.
Since $\mean{u}{x}$ is a monotone function of $u$, optimal value for $\score{U \mid u}$
will be smaller or equal than $\score{V \mid u}$.
\qed\end{proof}

\begin{lemma}
\label{lem:join}
Let $U$ and $V$ be two segments such that $\freq{U} < \freq{V}$.
Let $r$ and $t$ be two parameters, $r \geq t$.
Then there exists $u$, $r \geq u \geq t$ such that $\score{U \mid r} +  \score{V \mid t} \leq \score{U \mid u} + \score{V \mid u}$.
\end{lemma}

\begin{proof}
Let $r^*$ be such that $\score{U \mid r^*}$ is optimal, and define $t^*$ similarly.
Note that $r^*$ or $t^*$ can be infinite. Lemma~\ref{lem:concave} implies that
$r^* \leq t^*$.

Assume that $r \leq t^*$. 
This implies that $t \leq r \leq t^*$.
Since the score function is concave, $\score{V \mid r} \geq \score{V \mid t}$. Set $u = r$ to prove the lemma.

Assume that $r > t^*$. Set $u = \max\fpr{t^*, t}$.
Since $r \geq u \geq r^*$,
due to concavity of $\score{U \mid u} \geq \score{U \mid r}$ and by definition
$\score{V \mid u} \geq \score{V \mid t}$. This proves the lemma.
\qed\end{proof}

\begin{corollary}
\label{cor:monotonepar}
Let $\sgm{U}$ be a $K$-segmentation. Let $r_1, \ldots, r_K$ be $K$ parameters such that $r_i \geq r_{i - 1}$.
There exists a \emph{monotone} $K$-segmentation $\sgm{V} \preceq \sgm{U}$ such that
	$\score{\sgm{V}} \geq \score{\sgm{U} \mid r_1, \ldots, r_K}$.
\end{corollary}

\begin{proof}
Assume that $\sgm{U}$ is not monotone, that is, there exists $U_i$ and $U_{i +
1}$ such that $\freq{U_{i} \sm U_{i - 1}} < \freq{U_{i + 1} \sm
U_i}$.  Lemma~\ref{lem:join} implies that we can replace $r_i$ and $r_{i + 1}$
with a common parameter, say $u$, and not decrease the score.
We can now join the $i$th and $i + 1$th segments into one segment with a parameter $u$ and
add a new empty corner at the beginning of the segmentation (with an infinitely large parameter).
This modification makes the segmentation smaller w.r.t our order.
We repeat this step until convergence. When converged, we have obtained a monotone
$K$-segmentation whose score is at least as good as the original segmentation.
\qed\end{proof}

We will now show that we can modify a segmentation containing a non-border.

\begin{proposition}
Let 
$U_0, \ldots, U_K = \sgm{U}$ be a monotone $K$-segmentation. 
Assume that there exist $U_i$, $X$, and $Z$ such that $U_{i - 1} \subseteq X
\subsetneq U_i \subsetneq Z \subseteq U_{i + 1}$ and $\freq{Z \sm
U_i} \geq \freq{U_i \sm X}$.

Then there exists a monotone $K$-segmentation $\sgm{V}$ such that either $\sgm{V} \precneqq \sgm{U}$
and $\score{\sgm{V}} \geq \score{\sgm{U}}$ or $\score{\sgm{V}} > \score{\sgm{U}}$.
\label{prop:ascent}
\end{proposition}

In order to prove the proposition we will introduce some helpful notation. First, given
two parameters $r$ and $t$, we define
\[
	h(X ; r, t) = \score{X \sm U_{i - 1} \mid r} + \score{U_{i + 1} \sm X \mid t}\quad.
\]
We also define
\[
	 g(l, \delta ; r, t)  = l(Z(r) - Z(t) + (r - t)\delta)\quad.
\]
This function is essentially the difference between two scores.
\begin{lemma}
Let $U_{i - 1} \subseteq X \subsetneq Y \subseteq U_{i + 1}$.  We have $h(Y ; r, t) - h(X ; r, t) =g(\abs{Y \sm X}, \freq{Y \sm X} ; r, t)$.
\end{lemma}

\begin{proof}
Let us define $\cs{Z} = \sum_{z \in Z} D(z)$.
Note that
\[
\begin{split}
	h(X ; r, t) & = (\abs{X} - \abs{U_{i - 1}})Z(r) + r(\cs{X} - \cs{U_{i - 1}}) \\
	&\quad+ (\abs{U_{i + 1}} - \abs{X})Z(t) + t(\cs{U_{i + 1}} - \cs{X}) \\
	& = \abs{X}(Z(r) - Z(t)) + (r - t)\cs{X} + \text{const},
\end{split}
\]
where const does not depend on $X$. Let us write $d = Z(r) - Z(t)$ and $Z = Y \sm X$.
This allows us to write
\[
\begin{split}
	& h(Y ; r, t) - h(X ; r, t) \\
	&\quad = \abs{Y}d + (r - t)\cs{Y} - \abs{X}d - (r - t)\cs{X}  \\
	&\quad =\abs{Z}\big(d + (r - t)\frac{\cs{Z}}{\abs{Z}}\big) = g(\abs{Y \sm X}, \freq{Y \sm X} ; r, t)\quad.
\end{split}
\]
This completes the proof.
\qed\end{proof}

\begin{proof}[of Proposition~\ref{prop:ascent}]
Replace $U_i$ with $Z$ and let $\sgm{U}^*$ be the resulting segmentation. 
Similarly, replace $U_i$ with $X$ and let $\sgm{U}'$ be the resulting segmentation. 
Define $y = \sup_{r, t} h(U_i; r, t)$.

Fix $\epsilon > 0$.
Lemma~\ref{lem:concave} implies that there exist $r_1, \ldots, r_K$ s.t.
\[
	r_{j - 1} > r_j \text{ for } j = 2, \ldots, K  \quad\text{and}\quad \score{\sgm{U}^* \mid r_1, \ldots, r_K} \geq \score{\sgm{U}} - \epsilon\quad.
\]
From now on we will write $h(X)$ to mean $h(X ; r_i, r_{i + 1})$ and $g(k, \delta)$ to mean $g(k, \delta ; r_i, r_{i + 1})$.
Note that we must have $h(U_i) \geq y - \epsilon$.

Assume that $h(Z) > y + \epsilon$. Then $\score{\sgm{U}^* \mid r_1, \ldots, r_K} > \score{\sgm{U} \mid r_1, \ldots, r_K}  + \epsilon \geq \score{\sgm{U}}$.
Corollary~\ref{cor:monotonepar} applied to $\sgm{U}^*$ shows that there is a monotone segmentation
with a score better than $\score{\sgm{U}}$.

Assume that $h(Z) \leq y + \epsilon$.
We must have $h(U_i) + \epsilon \geq y \geq h(Z) - \epsilon$ or, equivalently, $2\epsilon \geq  h(Z) - h(U_i)$.

Define $\beta  = \freq{Z \sm U_i}$ and $\alpha  = \freq{U_i \sm X}$,
$n = \abs{Z \sm U_i}$, $m = \abs{U_i \sm X}$.
Define $c = n / m$.  We now have
\[
\begin{split}
	2\epsilon & \geq h(Z) - h(U_i) = g(n, \beta) = cg(m, \beta) \\
	& = cg(m, \alpha) + cm(r_i - r_{i + 1})(\beta - \alpha) \geq cg(m, \alpha) \\
	& = c(h(U_i) - h(X)) \geq c(y - \epsilon - h(X)), \\
\end{split}
\]
which implies $y - h(X) \leq \epsilon(1 + 2c^{-1}) \leq \epsilon(1 + 2\abs{U_K})$.
Corollary~\ref{cor:monotonepar} now implies that there exists a monotone segmentation $\sgm{V}$
with $\sgm{V} \precneqq \sgm{U}$ such that $y - \score{\sgm{V}}  \leq \epsilon(1 + 2\abs{U_K})$.
Since this holds for any $\epsilon > 0$, we have proved the proposition.
\qed\end{proof}

\begin{proof}[of Proposition~\ref{prop:borders}]
Assume that $U_j$ is not a border, then Lemma~\ref{lem:squeeze} implies that
there exist $U_i$, $X$, and $Y$ such that the conditions in
Proposition~\ref{prop:ascent} are satisfied. Apply
Proposition~\ref{prop:ascent} to obtain a new monotone segmentation, $\sgm{V}$.
Reapply the step to $\sgm{V}$ until $\sgm{V}$ consists only of borders. This procedure
terminates since at each step we either increase score or move segmentation is
moved to the left w.r.t the partial order $\prec$. There are finite number of
segmentations and no segmentation is visited twice, hence we converge to a segmentation
consisting only of borders.
\qed\end{proof}

\section{Proof of Proposition~\ref{prop:move}}\label{sec:propmove}

\begin{proof}[of Proposition~\ref{prop:move}]
Define $Z = \maxc{U}$. Let us first prove that $Z$ is a border.

Let $X = \minc{Z}$. If $U \sm X \neq \emptyset$, then
Lemma~\ref{lem:augment2} implies that $U$ is not a border.
Hence $U \subseteq X$. Lemma~\ref{lem:split} implies that
$\freq{Z \sm X} \geq \freq{Z \sm U}$.

Let $Y \supsetneq Z$ be a corner, then Lemma~\ref{lem:split}
implies that $\freq{Y \sm Z} < \freq{Z \sm U} \leq \freq{Z \sm X}$.
By definition, $Z \sm X$ has the smallest possible average.
Consequently, $Z$ is a border.

Assume that $Z \sm V \neq \emptyset$, then
Lemma~\ref{lem:augment} implies that $V$ is not a border,
which is a contradiction.  Hence $Z \subseteq V$. Since $V$
is the border next to $U$, we must have $Z = V$.

The proof in other direction is similar.
\qed\end{proof}

\section{Proof of Proposition~\lowercase{\ref{prop:flipcycle}}}\label{sec:proofflipcycle}
\vspace{1mm}

\begin{proof}
Let $k$ be as in Proposition~\ref{prop:converge}.
Fix $U_j \in \brd{T^k}$ with $j > 0$.
Let us write $V^i$ to be the portion of $T^i$
that corresponds to the entries in $C = U_j \sm U_{j - 1}$.
Since $\brd{T^i} = \brd{T^k}$ for $i > k$, it is enough
to prove the result by showing that there exists $m$ such that $V^m = V^{m + 2}$.
We will prove the result by induction over the size of $C$.

To that end, consider a DAG $G$ where the nodes are the entries in $C$. Two
distinct nodes $(a, b)$, $(c, d)$ are connected if and only if $a \leq c$ and
$b \leq d$.  Let $p$ be a sink in $G$. Define $V^k_p$ to be equal to
$V^k$ without the entry $p$ and let $V^i_p$, for $i > k$, be the order obtained
from $V^{i - 1}_p$ by simulating \textsc{FindOrder} with $w_F$.
Since $p$ is a sink, it does not block any entries as \findorder updates the order.
This implies that $V^i_p$ is equal to $V^i$ with $p$ deleted for any $i \geq k$.

By the induction assumption there exists $m_p$ such that $V^{m_p}_p = V^{m_p +
2}_p$ for each sink $p$.  Define $m = \max m_p$. Note that $V^m_p = V^{m + 2}_p$ for each sink $p$.

Assume now that we have only one sink, say $p$. Then $p$ will always
be last entry in $V^i$ and it follows immediately that $V^m = V^{m + 2}$.
We can safely assume that we have more than one sink.

Assume that we have exactly two sinks, say $p$ and $q$. Assume that $V^m_p$
does not end on $q$. Then it must be that $V^m_p$ ends on a parent of $p$ and
$p$ is the last entry in $V^m$. Since $V^m_p = V^{m + 2}_p$, the proposition
follows. Similar argument holds for $V^m_q$, $V^{m + 1}_p$, and $V^{m + 1}_q$.
Hence, we can safely assume that $V^m_p$, $V^m_q$, $V^{m + 1}_p$, and $V^{m + 1}_q$
all have $p$ or $q$ as their last entry. This implies that $V^m$, $V^{m + 1}$, and $V^{m + 2}$.
have $p$ and $q$ as their last entries. Assume that $p$ occurs before $q$ in $V^m$.
There will be a point during $\textsc{FindOrder}(w_F(\cdot; T^m))$
when $p$ and $q$ are in the heap $H$ at the same time, otherwise there will be an entry between $p$ and $q$ in $V^{m + 1}$.
This implies that $q$ occurs before $p$ in $V^{m + 1}$. We apply the same argument
to $V^{m + 1}$ to conclude that $p$ occurs before $q$ in $V^{m + 2}$. This implies
that $V^{m + 1} = V^{m + 2}$.

Assume that we have more than two sinks. Let $p$, $q$, and $r$ be three sinks.
We can deduce from $V^m_r$ whether $p$ occurs before $q$, or vice versa.  Since
we can do this for any sink triplet, we can deduce the order of sinks in $V^m$.
The last sink, say $p$, will be the last in $V^m$.  By definition of $m$, the
order will be the same in $V^{m + 2}$ and the same sink will be also the last
in $V^{m + 2}$. Since $V^m_p = V^{m + 2}_p$, the proposition follows.
\qed\end{proof}

\fi

\end{document}